\documentclass[aps,prx,reprint,superscriptaddress,nofootinbib,longbibliography]{revtex4-2}

\usepackage[T1]{fontenc}
\usepackage{amsmath}
\usepackage{amssymb}
\usepackage{amsthm}
\usepackage{graphicx}
\usepackage{booktabs}
\usepackage{microtype}
\usepackage[colorlinks=true,citecolor=blue,linkcolor=blue,urlcolor=blue]{hyperref}
\usepackage{physics}
\usepackage[utf8]{inputenc}
\usepackage{url}
\usepackage{enumitem}

\newcommand{\id}{\operatorname{id}}

\newcommand{\cM}{\mathcal M}
\newcommand{\cN}{\mathcal N}
\newcommand{\cP}{\mathcal P}
\newcommand{\cR}{\mathcal R}
\newcommand{\piL}{\pi_{\mathrm L}}
\newcommand{\dL}{d_{\mathrm L}}

\newcommand{\appref}[1]{Appendix\,\ref{#1}}

\newcommand{\eqnref}[1]{Eq.\,\eqref{#1}}

\newtheorem{theorem}{Theorem}
\newtheorem{proposition}{Proposition}
\newtheorem{lemma}{Lemma}
\newtheorem{corollary}{Corollary}[theorem]

\begin{document}
\makeatletter
\let\originaladdcontentsline\addcontentsline
\renewcommand{\addcontentsline}[3]{%
  \def\tempa{#1}%
  \def\tempb{toc}%
  \ifx\tempa\tempb 
  \else
    \originaladdcontentsline{#1}{#2}{#3}%
  \fi
}
\makeatother

\title{Hierarchy of R\'enyi Coherent Information in Stabilizer Codes}

\author{Akash Vijay}
\affiliation{Department of Physics and Institute for Condensed Matter Theory, University of Illinois Urbana-Champaign, Urbana, Illinois 61801, USA}

\author{Luis Colmenarez}
\affiliation{Institute for Quantum Information, RWTH Aachen University, 52056 Aachen, Germany}
\affiliation{Institute for Theoretical Nanoelectronics (PGI-2), Forschungszentrum Jülich, 52428 Jülich, Germany}

\author{Jong Yeon Lee}
\email[Contact Author:$~~$]{jongyeon@illinois.edu}
\affiliation{Department of Physics and Institute for Condensed Matter Theory, University of Illinois Urbana-Champaign, Urbana, Illinois 61801, USA}
\affiliation{Korea Institute for Advanced Study, Seoul 02455, South Korea}

\date{\today}

\begin{abstract}
R\'enyi coherent information, a computable proxy for the von Neumann coherent information, is widely used to study mixed-state phases of matter and decodability transitions in noisy quantum error-correcting codes. However, being a difference of two R\'enyi entropies, it need not be monotonic in the R\'enyi index, and lacks the operational meaning of its von Neumann counterpart.
Here we address both issues for stabilizer codes. First, for Pauli noise generated by independent Bernoulli events, we prove that the R\'enyi-$n$ coherent information is nondecreasing in $n \in \mathbb{Z}^+$. This follows from a general theorem: if independent random bits are mapped linearly to a fine label $T$ and a coarse label $C$, then the R\'enyi entropy difference $H_n(C)-H_n(T)$ is nondecreasing in $n$. For stabilizer codes, $T$ is the joint syndrome--logical class and $C$ is the syndrome, and the difference is the R\'enyi-$n$ coherent information up to a constant. 
The same theorem covers classical linear codes and independent detector error models. 
Second, for arbitrary stochastic Pauli noise, we give the R\'enyi-$n$ coherent information an operational meaning via postselection on matching syndromes between one data block and $n-1$ auxiliary blocks. 
We determine when this defines a quantum channel and show that saturation of the R\'enyi-$n$ coherent information is equivalent to asymptotically perfect recovery of the postselected channel. Moreover, the R\'enyi-$n$ coherent information also upper-bounds the ordinary coherent information achievable after any syndrome-conditioned recovery.

\end{abstract}

\maketitle
 
\section{Introduction}

R\'enyi information-theoretic quantities are widely used to study decoherence-induced phase transitions in noisy quantum systems~\cite{lee2022symmetry,Lee_2023,Fan2024,LeeCI2024,NiwaLee2025,kim2024errorthresholdsykcodes,LiMong2025,Lyons2024,Chen2024PRL, Wang2025, Sohal2025, yang2025topologicalmixedstatesphases, SuYangJian2024,vijay2025informationcriticalphasesdecoherence,vijay2025holographicallyemergentgaugetheory, PhysRevResearch.6.L042014, temkin2025chargeinformedquantumerrorcorrection, wang2025fractionalquantumhallstates, d4wh-hqcp}, because integer moments of the decohered density matrix often admit mappings to tractable statistical-mechanical models~\cite{DennisKitaevLandahlPreskill2002, Wang_2003, Bombin2012, Chubb_2021}.
A recurring numerical observation in these studies is that the error threshold inferred from the R\'enyi coherent information increases with the R\'enyi index.

However, this ordering is far from automatic: the R\'enyi coherent information is a difference of two R\'enyi entropies and therefore need not be monotonic in the R\'enyi index. 
Moreover, the von Neumann coherent information obeys a quantum data-processing inequality and has a direct operational connection to quantum error correction~\cite{PhysRevA.54.2614, SchumacherNielsen1996,PhysRevA.55.1613, BarnumKnill2002,BenyOreshkov2010}, whereas the corresponding difference of R\'enyi entropies does not generally inherit these properties.
At the same time, matching syndromes across $n$ copies weights each syndrome $s$ by $p_s^{n}$~\cite{LiMong2025}; decoding under this weighting reproduces finite-$n$ R\'enyi thresholds numerically~\cite{ColmenarezMartonMuller2026}, while the infinite-R\'enyi limit is related to pure post-selection thresholds~\cite{Smith_2024,English2025}.
These observations raise two distinct questions. First, under what conditions does the numerically observed ordering become a rigorous hierarchy in the R\'enyi index? Second, when does matched-syndrome postselection define a channel on an arbitrary logical input, and what recovery guarantees follow from the Rényi coherent information?

\begin{figure}[!t]
\includegraphics[width=0.48\textwidth]{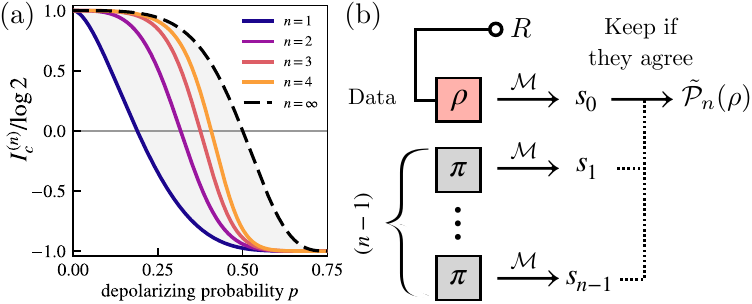}
\caption{\label{fig:overview}
{\bf (a) R\'enyi hierarchy.} Exact R\'enyi-$n$ coherent information for the $[[5,1,3]]$ code under depolarizing noise with total nonidentity probability $p$, with $q_X\,{=}\,q_Y\,{=}\,q_Z\,{=}\,p/3$. The hierarchy holds for every finite stabilizer code under independent-event Pauli noise, but can fail for more general Pauli channels. 
{\bf (b) Postselection gadget.} One arbitrary logical data state and $n-1$ maximally mixed logical ancilla states are encoded into separate blocks. Each block undergoes the same noise channel and syndrome measurement.
Retaining only matching syndrome outcomes implements the trace-nonincreasing map; when the success probability is input-independent, division by this probability yields a quantum channel.
}
\end{figure}


Here, we answer these questions.
First, we prove that the R\'enyi-$n$ coherent information $I_c^{(n)}$ of every finite qubit stabilizer code is nondecreasing with $n$ under Pauli noise generated by independent Bernoulli events; see Fig.~\ref{fig:overview}(a).
The result follows from a more general binary-linear theorem: for independent Bernoulli faults $b$ and linear labels $T(b)$ and $C(b)$, with the coarse label $C(b)$ determined by the fine label $T(b)$, the R\'enyi entropy difference $\Delta_n:=H_n(T(b))-H_n(C(b))$ is nonincreasing with $n$. For stabilizer codes, $T$ records the joint syndrome--logical Pauli class and $C$ the syndrome, giving $\Delta_n=\log\dL-I_c^{(n)}$, where $\dL$ is the logical dimension.
The theorem requires neither CSS structure nor locality and places no restriction on the logical dimension. It also applies to classical linear codes and to repeated-measurement and circuit-level detector models that admit an exact representation by independent Bernoulli faults. Each Bernoulli event may act on several qubits, so the theorem includes spatially correlated errors; without independence, the hierarchy can be violated (\appref{app:counterexamples}).

Second, building on matched-syndrome postselection~\cite{LiMong2025}, we combine one arbitrary logical input with $n-1$ maximally mixed logical auxiliary inputs, and postselect on matching syndromes. 
The normalized operation defines a quantum channel precisely when its success probability is input independent.
For arbitrary stochastic Pauli noise, this balance condition holds automatically.
At every fixed integer $n\geq2$, dimension-independent finite-size bounds make vanishing absolute R\'enyi deficit equivalent to asymptotically perfect recovery of the matched channel.
The saturated-phase boundary, when well defined, therefore coincides with the threshold for optimal decoding under power-weighted syndrome resampling~\cite{ColmenarezMartonMuller2026}.


\section{R\'enyi hierarchy}
\label{sec:binary-hierarchy}

We begin with a general result for independent binary faults and nested linear observations; the stabilizer-code hierarchy follows as a corollary.
For a finite classical random variable $X$ and $n>1$, the R\'enyi entropy is $H_n(X):=(1-n)^{-1}\log\sum_x\Pr(X=x)^n$.

Let $b=(b_1,\ldots,b_M)\in\mathbb F_2^M$ consist of independent bits $b_a\sim\operatorname{Bern}(p_a)$, and let $K_T\subseteq K_C\subseteq\mathbb F_2^M$ be nested subspaces.
The fine and coarse labels are the quotient classes
\begin{align} \label{eq:binary-labels}
T(b)&:=b+K_T\in\mathbb F_2^M/K_T, \nonumber \\
C(b)&:=b+K_C\in\mathbb F_2^M/K_C.
\end{align}
Thus $T(b)$ determines $C(b)$.
Equivalently, any linear maps with $\ker T\subseteq\ker C$ give this description by taking $K_T=\ker T$ and $K_C=\ker C$.
Define the entropy difference
\begin{align}
\Delta_n:=H_n(T(b))-H_n(C(b)).
\label{eq:binary-deficit}
\end{align}
Since $C(b)$ is determined by $T(b)$, the Shannon-order difference is $\Delta_1=H(T(b)|C(b))$, the residual uncertainty in the fine label at a fixed coarse label.

\begin{theorem}[Binary-linear R\'enyi hierarchy]
\label{thm:binary-linear-hierarchy}
For every finite binary-linear model defined above,
\begin{align}
\Delta_1\geq\Delta_2\geq\Delta_3\geq\cdots\geq\Delta_\infty.
\label{eq:binary-linear-hierarchy}
\end{align}
\end{theorem}
We now connect this theorem to stabilizer codes. For a density matrix, the R\'enyi entropy is defined analogously from its eigenvalue spectrum. Let $\rho_{RQ}$ be the state obtained by encoding one half of a maximally entangled logical state and applying the noise, with $R$ the reference and $Q$ the physical output. Its R\'enyi coherent information is
\begin{align}
I_c^{(n)}:=H_n(Q)-H_n(RQ),
\end{align}
with $n=1$ giving the von Neumann coherent information.

\begin{corollary}[Stabilizer-code hierarchy]
\label{cor:hierarchy}
For any finite qubit stabilizer code under independent-event Pauli noise, the R\'enyi coherent information obeys
\begin{align}
I_c^{(1)} \leq I_c^{(2)} \leq I_c^{(3)} \leq \cdots \leq I_c^{(\infty)}.
\label{eq:pointwise-hierarchy}
\end{align}
\end{corollary}
Complementing a fault bit translates both labels by constants without changing their entropies.
It therefore suffices to prove the theorem for $0<p_a\leq1/2$; deterministic faults follow by continuity.
The proof has two steps: a random-rank sandwich establishes $\Delta_1\geq\Delta_2$, and a ferromagnetic pinning argument compares successive integer orders $n\geq2$. 

\begin{proof}[Step I: Shannon to collision]
For a set of fault locations $\Gamma$, let $T_\Gamma$ and $C_\Gamma$ be the restrictions to patterns supported on $\Gamma$.
Define
\begin{align}
\kappa(\Gamma)
&:=\dim T_\Gamma(\ker C_\Gamma) =\operatorname{rank}T_\Gamma-\operatorname{rank}C_\Gamma.
\label{eq:hidden-rank-main}
\end{align}
This counts the independent fine-label variations that leave the coarse label unchanged.  It is nondecreasing under inclusion of $\Gamma$.
Set $\theta_a:=4p_a(1-p_a)$, and let $\mu$ include each location $a$ in $\Gamma$ independently with probability $\theta_a$.
We prove that
\begin{align}
\Delta_1\geq(\log2)\mathbb E_\mu\kappa(\Gamma)\geq\Delta_2.
\label{eq:random-rank-sandwich}
\end{align}

For the left inequality of \eqnref{eq:random-rank-sandwich}, choose $x$ uniformly at random from $\mathbb F_2^M$, independently of $b$, and set $y=x+b$. 
Then each map $x_a \mapsto y_a$ is a binary symmetric channel with crossover probability $p_a$.
Uniform randomness of $x$ makes $y$ independent of $b$. 
At fixed $y$, the relation translates both labels by known constants, giving
\begin{align}
\Delta_1=H(T(b)|C(b))=H(T(x)|C(x),y).
\end{align}
For an arbitrary joint distribution of $U$ and a binary input $X_a$, let $Y_a$ and $Y_{\mathrm{er},a}$ be the outputs of a binary symmetric channel with crossover probability $p_a$ and a binary erasure channel with erasure probability $\theta_a$, respectively.
Since $1-\theta_a=(1-2p_a)^2$, the symmetric-channel contraction bound and the erasure identity give~\cite{Ahlswede1976,KornerMarton1977Comparison}
\begin{align}
I(U;Y_a)
&\leq(1-2p_a)^2I(U;X_a)=I(U;Y_{\mathrm{er},a}).
\end{align}
This comparison also holds for product channels with arbitrary, possibly correlated input coordinates~\cite{raginsky2016strongdataprocessinginequalities,PolyanskiyWu2017,Makur_2018}.
Let $y_{\mathrm{er}}$ be obtained from $x$ by erasing coordinate $a$ with probability $\theta_a$, independently across coordinates and independently of $x$.
Applying the product-channel comparison to the conditional distribution for $x$ given $C(x)=c$, with auxiliary variable $U=T(x)$, gives
\begin{align}
I(T(x);y_{\mathrm{er}}|C(x)=c)
\geq
I(T(x);y|C(x)=c).
\end{align}
As $H(T|C,y){=}H(T|C){-}I(T;y|C)$, averaging over $c$ gives
\begin{align}
\Delta_1
=H(T(x)|C(x),y)
\geq
H(T(x)|C(x),y_{\mathrm{er}}).
\label{eq:bec-conditional-bound}
\end{align} 
Let $\Gamma$ be the erased set, so $\Gamma\sim\mu$.
Given the erasure output $y_{\mathrm{er}}$, the values of $x_a$ for $a\notin\Gamma$ are revealed exactly, whereas the coordinates $x_a$ for $a\in\Gamma$ are unobserved. 
Fixing $C(x)$ leaves the remaining vector uniformly distributed on a translate of $\ker C_\Gamma$.
Since every attainable fine label has the same number of compatible inputs, $T(x)$ is uniform on a translate of $T_\Gamma(\ker C_\Gamma)$, which has dimension $\kappa(\Gamma)$. Therefore, averaging over $\Gamma$ gives
\begin{align}
\Delta_1 \geq (\log2)\mathbb E_\mu\kappa(\Gamma).
\label{eq:delta1-rank-main}
\end{align}

For the right inequality, let $b'$ be an independent copy of $b$ and set $w:=b+b'$. The two realizations have the same fine label when $T(w)=0$, and the same coarse label when $C(w)=0$. 
Since $e^{-H_2(X)}$ is the collision probability of $X$ and $\ker T \subseteq \ker C$, 
\begin{align}
e^{-\Delta_2}= \frac{\Pr[T(w)=0]}{\Pr[C(w)=0]}= \Pr[T(w)=0|C(w)=0].
\label{eq:delta2-collision-ratio}
\end{align}
Each $w_a$ equals one with probability $\theta_a/2$.
Equivalently, draw $\Gamma\sim\mu$, set $w_a=0$ outside $\Gamma$, and sample the coordinates inside $\Gamma$ independently and uniformly from $\{0,1\}$.
Rank-nullity gives $\Pr[C(w)=0|\Gamma]=2^{-\operatorname{rank}C_\Gamma}$ and the analogous expression for $T$.
Averaging these probabilities and using Eq.~\eqref{eq:hidden-rank-main} yields
\begin{align}
e^{-\Delta_2}
=\frac{\mathbb E_\mu\left[2^{-\operatorname{rank}C_\Gamma}2^{-\kappa(\Gamma)}\right]}
{\mathbb E_\mu 2^{-\operatorname{rank}C_\Gamma}}.
\label{eq:delta2-rank-factorized}
\end{align}
Since both factors in the numerator are nonincreasing as $\Gamma$ grows, Harris's correlation inequality~\cite{Harris1960} gives
\begin{align}
\mathbb E_\mu \Big[2^{-\operatorname{rank}C_\Gamma}2^{-\kappa(\Gamma)} \Big] \geq \mathbb E_\mu \Big[2^{-\operatorname{rank}C_\Gamma} \Big] \mathbb E_\mu \Big[ 2^{-\kappa(\Gamma)}\Big].
\end{align}
Combining this with Eq.~\eqref{eq:delta2-rank-factorized} and Jensen's inequality gives $e^{-\Delta_2}\geq\mathbb E_\mu2^{-\kappa(\Gamma)}\geq2^{-\mathbb E_\mu\kappa(\Gamma)}$.
Taking negative logarithms proves the right inequality in Eq.~\eqref{eq:random-rank-sandwich} and completes Step~I.
\end{proof}

\begin{proof}[Step II: higher integer orders]
We prove $\Delta_{m+2}\leq\Delta_{m+1}$ for every integer $m\geq1$.
Let $q(b)=\prod_aq_a(b_a)$ be the product Bernoulli law, with $q_a(1)\,{=}\,p_a$.
For a subspace $H$, let $P_H(\mathcal C):=\sum_{b\in\mathcal C}q(b)$ be its induced distribution on $\mathbb F_2^M/H$.
The $(m+1)$st moment is
\begin{align}
\mathcal M_{m+1}(H)
:=\sum_{\mathcal C\in\mathbb F_2^M/H}P_H(\mathcal C)^{m+1}.
\label{eq:direct-coset-moment-def}
\end{align}
Taking the first configuration of each ordered $(m+1)$-tuple as $b$, every tuple in one coset is uniquely represented as $(b,b+h_1,\ldots,b+h_m)$ with $h_r\in H$.
Writing $\mathbf h:=(h_1,\ldots,h_m)$ and $W_m(b,\mathbf h):=q(b)\prod_{r=1}^mq(b+h_r)$, the moment expansion gives
\begin{align}
\mathcal M_{m+1}(H)
=\sum_{b\in\mathbb F_2^M}\sum_{\mathbf h\in H^m}W_m(b,\mathbf h).
\label{eq:direct-coset-moment}
\end{align}
The fine and coarse labels correspond to $H=K_T$ and $H=K_C$, respectively, so
\begin{align}
\Delta_{m+1}
=\frac{1}{m}\log\frac{\mathcal M_{m+1}(K_C)}{\mathcal M_{m+1}(K_T)}.
\label{eq:direct-delta-moment-ratio}
\end{align}

We next express the moment as a ferromagnetic Ising partition function. Set $s_a:=(-1)^{b_a}$, $\chi_a(h):=(-1)^{h_a}$, and $J_a:=\frac12\log\frac{1-p_a}{p_a}\geq0$.
The sum over the reference configuration $b$ in Eq.~\eqref{eq:direct-coset-moment} makes the $s_a$ dynamical spins.
Since $q_a(b_a)\propto e^{J_as_a}$,
\begin{align}
 W_m(b, \mathbf h) \propto e^{ \sum_aJ_as_a\left[1+\sum_{r=1}^m\chi_a(h_r)\right] }.
\label{eq:direct-replica-ferromagnet}
\end{align}
To make the Ising structure explicit, choose a linear complement $L$ of $K_T$ in $K_C$,
\begin{align}
K_C=K_T\oplus L,\quad d_{\mathrm{hid}}:=\dim L=\dim(K_C/K_T).
\end{align}
Let $d_T:=\dim K_T$, with $\{g_i\}_{i=1}^{d_T}$ being a basis of $K_T$ and $\{\ell_j\}_{j=1}^{d_{\mathrm{hid}}}$ a basis of $L$, so their union forms a basis of $K_C$.
Every $h_r\in K_C$ has a unique decomposition
\begin{align}
h_r=\sum_{i=1}^{d_T}x_{r,i}g_i+\sum_{j=1}^{d_{\mathrm{hid}}}\lambda_{r,j}\ell_j,\qquad x_{r,i},\lambda_{r,j}\in\mathbb F_2.
\end{align}
The spins $\sigma_{r,i}:=(-1)^{x_{r,i}}$ and $\zeta_{r,j}:=(-1)^{\lambda_{r,j}}$ give
\begin{align}
\chi_a(h_r)=\prod_{i:g_{i,a}=1}\sigma_{r,i}\prod_{j:\ell_{j,a}=1}\zeta_{r,j}.
\end{align}
Every interaction in Eq.~\eqref{eq:direct-replica-ferromagnet} is therefore an Ising-spin monomial with a nonnegative coupling.

Note that $h_r\in K_T$ exactly when all $\zeta_{r,j}=+1$.
To pin these components, introduce $t\geq0$ and define
\begin{align}
\mathcal Z_{m,t}:=\sum_{b\in\mathbb F_2^M}\sum_{\mathbf h\in K_C^m}W_m(b,\mathbf h) e^{ t\sum_{r=1}^m\sum_{j=1}^{d_{\mathrm{hid}}}(\zeta_{r,j}-1) }.
\label{eq:pinned-partition-function}
\end{align}
At $t=0$, $\mathcal Z_{m,0}=\mathcal M_{m+1}(K_C)$.
As $t\to\infty$, only configurations with every $h_r\in K_T$ survive, so $\mathcal Z_{m,t}\to\mathcal M_{m+1}(K_T)$.
With $F_m(t):=-\log\mathcal Z_{m,t}$, Eq.~\eqref{eq:direct-delta-moment-ratio} becomes
\begin{align}
\Delta_{m+1}=\frac{F_m(\infty)-F_m(0)}{m},
\end{align}
the pinning free-energy cost per nonreference replica.
Writing $\langle\cdot\rangle_{m,t}$ for expectation under the normalized pinned weight, permutation symmetry among the $m$ nonreference replicas gives $\frac{d}{dt} F_m(t)=m\sum_j[1-\langle\zeta_{1,j}\rangle_{m,t}]$.
Integrating yields
\begin{align}
\Delta_{m+1}
=\sum_{j=1}^{d_{\mathrm{hid}}}\int_0^\infty
\left[1-\langle\zeta_{1,j}\rangle_{m,t}\right]\mathrm{d}t.
\label{eq:direct-pinning-integral}
\end{align}
It therefore suffices to prove $\langle\zeta_{1,j}\rangle_{m+1,t}\geq\langle\zeta_{1,j}\rangle_{m,t}$.
We use the following correlation inequality.
\begin{lemma}[GKS inequality~\cite{Griffiths,KellySherman1968,Ginibre1970}]
\label{lem:GKS}
Consider finitely many Ising spins $\sigma_i=\pm1$ with probability weight
\begin{align}
P(\sigma) \propto \exp\bigg[\sum_AJ_A\sigma_A \bigg], 
\qquad \sigma_A:=\prod_{i\in A}\sigma_i,
\end{align}
with all $J_A\geq0$. Then any two spin monomials satisfy
\begin{align}
\frac{\partial}{\partial J_B}\langle\sigma_A\rangle
=\operatorname{Cov}(\sigma_A,\sigma_B)\geq0.
\end{align}
\end{lemma}

The key observation is that $\mathcal Z_{m+1,t}$ can be reduced to $\mathcal Z_{m,t}$, up to a decoupled multiplicative factor, by continuously turning off the coupling of the last variable $h_{m+1}$ to the reference spins.
Introduce $\eta \in[0,1]$ and replace the exponent in Eq.~\eqref{eq:direct-replica-ferromagnet} by
\begin{align}
\sum_aJ_as_a\left[1+\sum_{r=1}^m\chi_a(h_r)+ \eta \chi_a(h_{m+1})\right],
\label{eq:add-replica-interpolation}
\end{align}
while retaining the same pinning field $t$ on all $m+1$ variables.
Let $\langle\cdot\rangle_{m+1,t,\eta}$ denote expectation in this interpolating ensemble.
At $\eta=1$, we recover $\mathcal Z_{m+1,t}$.
At $\eta=0$, the last variable $h_{m+1}$ factorizes, so normalized expectations involving the first $m$ variables coincide with those in $\mathcal Z_{m,t}$.
In particular,
\begin{align}
\langle\zeta_{1,j}\rangle_{m+1,t,0}&=\langle\zeta_{1,j}\rangle_{m,t}, \\
\langle\zeta_{1,j}\rangle_{m+1,t,1}&=\langle\zeta_{1,j}\rangle_{m+1,t}.
\end{align}
All couplings in \eqnref{eq:add-replica-interpolation} remain nonnegative, including the pinning fields $t\zeta_{r,j}$.
Both $\zeta_{1,j}$ and $s_a\chi_a(h_{m+1})$ are spin monomials, so Lemma~\ref{lem:GKS} gives
\begin{align}
\frac{\partial}{\partial \eta}\langle\zeta_{1,j}\rangle_{m+1,t,\eta}
&=\sum_aJ_a\operatorname{Cov}_{m+1,t,\eta}\left(\zeta_{1,j},s_a\chi_a(h_{m+1})\right) \nonumber \\
& \geq 0.
\label{eq:direct-GKS-derivative}
\end{align}
Comparing endpoints yields $\langle\zeta_{1,j}\rangle_{m+1,t}\geq\langle\zeta_{1,j}\rangle_{m,t}$. Substitution into Eq.~\eqref{eq:direct-pinning-integral} proves $\Delta_{m+2}\leq\Delta_{m+1}$ for every $m\geq1$.
Together with Step~I, this establishes all finite integer orders.
Since the system is finite, taking $n\to\infty$ gives the infinite-order endpoint and completes the proof of Theorem~\ref{thm:binary-linear-hierarchy}.
\end{proof}

 \textit{Coherent information in stabilizer codes.--}
For an $N$-qubit stabilizer code encoding $k$ qubits, write $\dL=2^k$.
Under stochastic Pauli noise, let $S$ and $L$ denote the syndrome and logical Pauli class, using a fixed choice of syndrome representatives.
Orthogonal Bell sectors give $H_n(RQ)=H_n(S,L)$~\cite{LeeCI2024,NiwaLee2025}, while the maximally mixed physical marginal within each syndrome sector gives $H_n(Q)=H_n(S)+\log\dL$.
Consequently,
\begin{align}
I_c^{(n)}&=\log\dL-\Delta_n,\quad \Delta_n=H_n(S,L)-H_n(S).
\label{eq:stabilizer-binary-identification}
\end{align}
Both the joint state and its physical marginal are already block diagonal in the syndrome sectors.
Explicitly measuring and retaining the syndrome therefore leaves these entropies unchanged.

Now assume independent Bernoulli Pauli noise events $b$.
Represent Paulis modulo phase by $\mathsf P_N=\mathbb F_2^{2N}$ with its symplectic pairing~\cite{681315}.
Let $\mathsf S\subseteq\mathsf P_N$ be the stabilizer subspace and $\mathsf N:=\mathsf S^{\perp_{\mathrm{sp}}}$ its Pauli normalizer modulo phase.
For event vectors $e_a$, the realized Pauli is $e(b):=\sum_ab_ae_a$.
The labels
\begin{align}
T(b)&:=e(b)+\mathsf S, \qquad C(b):=e(b)+\mathsf N
\end{align}
are the joint syndrome--logical class and the syndrome, respectively.
Since $\mathsf S\subseteq\mathsf N$, their kernels are nested, and Theorem~\ref{thm:binary-linear-hierarchy} proves Corollary~\ref{cor:hierarchy}.
Here $d_{\mathrm{hid}}=\dim(\ker C/\ker T)\leq2k$, with equality exactly when zero-syndrome combinations of the allowed events span all logical Pauli classes. The independent-event assumption cannot be dropped in general: \appref{app:counterexamples} gives stochastic Pauli channels for which the hierarchy fails.
 
\textit{Detector error models and classical codes.--}
For an exact independent detector error model, let $D(b)$ be the detector pattern and $\Lambda(b)$ the logical-observable flips, both linear in the independent Bernoulli faults~\cite{Gidney_2021}.
Taking $C(b):=D(b)$ and $T(b):=(D(b),\Lambda(b))$ gives $\Delta_n=H_n(D,\Lambda)-H_n(D)$, which obeys the same hierarchy.
This includes repeated-measurement and circuit-level models whenever their fault mechanisms admit an exact independent Bernoulli representation.
Independent Pauli faults propagated through Clifford circuits and independent measurement or reset flips give such binary-linear responses.
Standard depolarizing Pauli noise is included up to complete depolarization.
For a classical binary linear code under independent bit flips, take $T(b)=b$ and $C(b)=Hb$, where $H$ is the parity-check matrix.

Both proof steps use the binary structure.
Although the coset-moment identity extends to finite abelian groups, the higher-order argument uses real $\pm1$ characters and nonnegative Ising couplings; a nonbinary extension therefore requires a different correlation argument.

\section{Matched-syndrome operation}
\label{sec:matched-syndrome}

We now relate the R\'enyi deficit to recovery after matched-syndrome postselection. We allow arbitrary stochastic Pauli noise, without the independent-event assumption of Sec.~\ref{sec:binary-hierarchy}. Let $V:\mathcal H_{\mathrm L}\to\mathcal H_Q$ encode a stabilizer code of logical dimension $\dL$, and let $\Pi_s$ project onto its $\dL$-dimensional syndrome sector $s$. The unnormalized branches of the noise $\cN$ are
\begin{align}
\cM_s(\rho_{\mathrm L}):=\Pi_s\cN(V\rho_{\mathrm L}V^\dagger)\Pi_s.
\label{eq:branch-map}
\end{align}
For Pauli noise, $p_s:=\Tr\cM_s(\rho_{\mathrm L})$ is independent of $\rho_{\mathrm L}$. Write $q_s(\ell):=\Pr(L=\ell\mid S=s)$ for the logical Pauli posterior, omitting branches with $p_s=0$.
For a normalized maximally entangled input $\Phi_{R\mathrm L}$, the conditional Choi states are
\begin{align}
J_s:=\frac{(\id_R\otimes\cM_s)(\Phi_{R\mathrm L})}{p_s}
=\sum_\ell q_s(\ell)\ketbra{\Phi_{s,\ell}}.
\label{eq:syndrome-resolved-choi}
\end{align}
Here $\ket{\Phi_{s,\ell}}$ is the logical Bell state labeled by $\ell$, embedded in syndrome sector $s$ using a fixed Pauli representative for $s$.
Recording $s$ in a classical register $S$ gives $\Omega_{RSQ}:=\sum_sp_s\ketbra{s}_S\otimes J_s$, whose R\'enyi coherent information and deficit are $I_c^{(n)}$ and $\Delta_n$ of Eq.~\eqref{eq:stabilizer-binary-identification}.

\emph{Matched-syndrome operation.--}
For $n\geq2$, prepare one block in an arbitrary logical state $\rho_{\mathrm L}$ and $n-1$ auxiliaries in $\piL:=I_{\mathrm L}/\dL$.
Encode each block, apply the noise and syndrome measurement independently, and retain the data block and common syndrome only when all outcomes agree; see Fig.~\ref{fig:overview}(b).
The auxiliaries all report $s$ with probability $p_s^{n-1}$, giving the unnormalized operation
\begin{align}
\cP_n(\rho_{\mathrm L}):=\sum_sp_s^{n-1}\ketbra{s}_S\otimes\cM_s(\rho_{\mathrm L}).
\label{eq:postselection-operation}
\end{align}
We call the matched-syndrome operation \emph{$n$-balanced} if its acceptance probability is independent of the logical input~\cite{nbalance}.
For any stabilizer code under stochastic Pauli noise, each syndrome probability $p_s$ is input independent, so $n$-balance holds automatically for every integer $n\geq2$.
The acceptance probability is $Z_n:=\sum_sp_s^n=e^{-(n-1)H_n(S)}$, and hence $\widetilde{\cP}_n:=\cP_n/Z_n$ is a quantum channel with Choi state
\begin{align}
\Theta_n:=(\id_R\otimes\widetilde{\cP}_n)(\Phi_{R\mathrm L})
=\sum_sQ_n(s)\ketbra{s}_S\otimes J_s,
\label{eq:flagged-escort-choi}
\end{align}
where $Q_n(s):=p_s^n/Z_n$ is the normalized power-weighted distribution. Matching therefore changes the syndrome weights but not the conditional states $J_s$.
When $H_n(S)$ is extensive, $Z_n$ is exponentially small in system size, making the implementation exponentially costly.

\begin{proposition}[Matched-syndrome recovery under Pauli noise]
\label{prop:pauli-recovery}
Let $\epsilon_n:=1-\sum_sQ_n(s)\max_\ell q_s(\ell)$ be the maximum-likelihood error probability under $Q_n$.
For every syndrome-conditioned recovery channel $\cR:SQ\to\mathrm L$,
\begin{align}
I(R\rangle\mathrm L)_{(\id_R\otimes\cR)(\Theta_n)}
\leq I_{c,\mathrm{match}}\leq I_c^{(n)},
\label{eq:matched-ci-data-processing}
\end{align}
where $I_{c,\mathrm{match}}:=I(R\rangle SQ)_{\Theta_n}$ is the ordinary coherent information of $\widetilde{\cP}_n$ at input $\piL$.
The optimal recovered entanglement fidelity is $F_n=1-\epsilon_n$, and
\begin{align}
F_n^n\leq e^{-(n-1)\Delta_n}\leq F_n.
\label{eq:simpler_bound}
\end{align}
\end{proposition}

\begin{proof}
Equation~\eqref{eq:stabilizer-binary-identification} gives
\begin{align}
B_n:=e^{-(n-1)\Delta_n} =\sum_sQ_n(s)\sum_\ell q_s(\ell)^n.
\label{eq:matched-posterior-moment}
\end{align}
Jensen's inequality and $H_n\,{\leq}\,H$ imply $\Delta_n\leq\sum_sQ_n(s)H_n(q_s)\leq\sum_sQ_n(s)H(q_s)$.
Since $\Theta_n$ has Bell-diagonal blocks, $I_{c,\mathrm{match}}=\log\dL-\sum_sQ_n(s)H(q_s)$, proving the right inequality in Eq.~\eqref{eq:matched-ci-data-processing}.
The left inequality is ordinary coherent-information data processing~\cite{SchumacherNielsen1996}.

Set $r_s:=\max_\ell q_s(\ell)$.
For any recovery channel $\cR_s$ on sector $s$, the operator $X_s:=(\id_R\otimes\cR_s^\dagger)(\Phi_{R\mathrm L})$ is positive and has unit trace, since recovery is trace preserving and the sector has dimension $\dL$.
Hence the recovered entanglement fidelity $\Tr(J_sX_s)$ is at most $r_s$, attained by correcting a most likely logical Pauli.
Optimizing separately for each syndrome gives $F_n=\sum_sQ_n(s)r_s=1-\epsilon_n$.
For $n\geq2$, the pointwise bounds $r_s^n\leq\sum_\ell q_s(\ell)^n\leq r_s$ give Eq.~\eqref{eq:simpler_bound} after averaging over $Q_n$ and using Jensen's inequality $\sum_sQ_n(s)r_s^n\geq F_n^n$.
\end{proof}

Therefore, Eq.~\eqref{eq:matched-ci-data-processing} implies that $I_c^{(n)}\leq0$ rules out positive ordinary coherent information after any syndrome-conditioned recovery of $\widetilde{\cP}_n$, while Eq.~\eqref{eq:simpler_bound} gives $F_n\leq\dL^{-(n-1)/n}$.
Positive $I_c^{(n)}$ can nevertheless coexist with $I_{c,\mathrm{match}}<0$~\cite{trivial-qubit-check}. 
For $n\geq2$, Eq.~\eqref{eq:simpler_bound} makes vanishing absolute deficit $\Delta_n\to0$ equivalent to $\epsilon_n\to0$, even for code families with growing $\dL$.
Saturation therefore characterizes asymptotically perfect recovery of the full logical block, and the saturated-phase boundary is the decoding threshold under $Q_n$.

The implication from saturation to recovery extends to general quantum noise: Proposition~\ref{prop:recovery} of Appendix~\ref{app:recovery-bounds} establishes $F_n\geq e^{-2(n-1)\Delta_n}$ for the syndrome-resolved state, preserving the dimension-independent implication from saturation to recovery at fixed $n$.
It also proves the converse at fixed $\dL$.
Under $n$-balance, these state-recovery bounds apply to the matched channel (Appendix~\ref{app:balance}).


\section{Conclusion and Outlook}

We established a finite-size binary-linear hierarchy for independent Bernoulli faults: the R\'enyi entropy difference between a fine linear label and its linear coarse-graining is nonincreasing with the integer order.
This yields the coherent-information hierarchy for stabilizer codes and also applies to classical linear codes and exact independent detector error models.
For one-parameter noise families in this class with a single positive-to-negative zero crossing at each order, the crossing strengths $p_c^{(n)}$ satisfy
\begin{align}
p_c^{(1)}\leq p_c^{(2)}\leq p_c^{(3)}\leq\cdots\leq p_c^{(\infty)}.
\end{align}

Matched-syndrome postselection gives a complementary operational interpretation for arbitrary stochastic Pauli noise, where $n$-balance holds automatically.
Given $n$, absolute saturation $\Delta_n\to0$ is equivalent to asymptotically perfect maximum-likelihood decoding under $Q_n(s)\propto p_s^n$. The saturated-phase boundary, when well defined, is therefore the recovery threshold of the matched channel.
The zero crossing gives a distinct, one-way certificate: at the maximally mixed logical input, $I_c^{(n)}\leq0$ rules out positive ordinary coherent information after any syndrome-conditioned recovery of this channel.

The independent-event assumption is sufficient but may not be optimal.
The counterexample in \appref{app:counterexamples} shows that product Pauli noise can violate the hierarchy even for a distance-two CSS code.
This example is inhomogeneous, leaving open whether homogeneous product noise $\nu^{\otimes N}$ can produce such a violation.
Characterizing the broadest noise class supporting the hierarchy and determining whether monotonicity extends to noninteger orders are natural future directions. 
A common formulation of the random-rank and ferromagnetic pinning arguments may clarify equality cases and the role of Fourier positivity, and suggest extensions beyond binary-linear systems, including to qudits and suitable coherent-noise models. 
More broadly, for general quantum channels,  integer-R\'enyi coherent information may remain experimentally accessible through randomized measurements even when its von Neumann counterpart is not. This raises the possibility of reconstructing the latter---and its associated threshold---by analytic continuation~\cite{Vijay_SAC}.

Beyond threshold ordering, it would be valuable to understand how syndrome reweighting changes the critical behavior of recovery transitions.
The matched-syndrome construction also motivates studying the tradeoff between acceptance probability and conditional recovery fidelity.
In particular, can less restrictive conditioning on syndrome outcomes retain quantitative recovery guarantees at lower postselection cost~\cite{English2025}?
Understanding these questions could clarify which features of replica phase diagrams reflect recoverability under the original noise and which depend on the conditioned syndrome ensemble.


\section*{acknowledgments}

A. V. and J. Y. L. are supported by the faculty startup grant at the University of Illinois, Urbana-Champaign, and the IBM-Illinois Discovery Accelerator Institute. 
L. C. is supported by the U.S. ARO Grant No. W911NF-21-1-0007, IARPA and the Army Research Office under the Entangled Logical Qubits program, and Cooperative Agreement Number W911NF-23-2-0216. 
The views and conclusions contained in this document are those of the authors and should not be interpreted as representing the official policies, either expressed or implied, of IARPA, the Army Research Office, or the U.S. Government. 
The U.S. Government is authorized to reproduce and distribute reprints for Government purposes notwithstanding any copyright notation herein. 
J. Y. L. is supported by the Quantum Universe Center scholar program at KIAS.

\clearpage 
\newpage

\makeatletter
\let\addcontentsline\originaladdcontentsline
\makeatother

\appendix

\makeatletter
\@removefromreset{equation}{section}
\@removefromreset{figure}{section}
\@removefromreset{table}{section}
\makeatother
\onecolumngrid

\makeatletter
\def\l@subsection#1#2{}
\def\l@subsubsection#1#2{}
\makeatother

\setcounter{equation}{0}
\setcounter{figure}{0}
\setcounter{table}{0}
\setcounter{theorem}{0}
\setcounter{proposition}{0}
\setcounter{lemma}{0}
\setcounter{corollary}{0}

\makeatletter
\renewcommand{\theequation}{S\arabic{equation}}
\renewcommand{\thefigure}{S\arabic{figure}}
\renewcommand{\thetable}{S\arabic{table}}
\renewcommand{\thetheorem}{S\arabic{theorem}}
\renewcommand{\theproposition}{S\arabic{proposition}}
\renewcommand{\thelemma}{S\arabic{lemma}}
\renewcommand{\thecorollary}{S\arabic{corollary}}
\setcounter{subsection}{0}
\makeatother

\begin{center}
{\large\bfseries Supplemental Material for}\\[1ex]
{\large\bfseries ``Hierarchy of R\'enyi Coherent Information in Stabilizer Codes''}\\[2ex]
\end{center}

\medskip

The Supplemental Material collects counterexamples, properties of the matched-syndrome operation, recovery bounds, and limitations of the recovery interpretation. We use the notation of the main text throughout; Choi states and maximally entangled states are normalized, and all logarithms are natural.

\tableofcontents 

\section{Counterexamples beyond the independent-event class}
\label{app:counterexamples}

We construct two product Pauli channels for which the hierarchy fails.
The first violates $\Delta_1\geq\Delta_2$; the second satisfies this inequality but violates $\Delta_2\geq\Delta_3$.
A single-qubit Fourier criterion shows that both channels lie outside the independent-event class.
Throughout this appendix, single-qubit Pauli probability vectors are ordered as $(I,X,Y,Z)$.

\subsection{A Fourier criterion for the independent-event class}

Let $\nu$ be a probability distribution on the single-qubit Pauli group modulo phase, $\mathsf P_1=\{I,X,Y,Z\}$.
Its symplectic Fourier coefficients are
\begin{align}
\widehat\nu(u):=\sum_{E\in\mathsf P_1}\nu(E)(-1)^{[u,E]}, \qquad u \in \mathsf{P}_1
\end{align}
where $[u,E]\in\mathbb F_2$ is zero when $u$ and $E$ commute and one when they anticommute.
If the realized Pauli is the product of independent events $e_a$ occurring with probabilities $p_a$, independence gives
\begin{align}
\widehat\nu(u)
&=\prod_a\left[(1-p_a)+p_a(-1)^{[u,e_a]}\right] = \prod_{a:[u,e_a]=1}(1-2p_a).
\label{eq:fourier-product}
\end{align}
Grouping events by Pauli type, define $\alpha:=\prod_{a:e_a=X}(1-2p_a)$, $\beta:=\prod_{a:e_a=Y}(1-2p_a)$, and $\gamma:=\prod_{a:e_a=Z}(1-2p_a)$.
Then $\alpha,\beta,\gamma\in[-1,1]$ and
\begin{align}
\bigl(\widehat\nu(X),\widehat\nu(Y),\widehat\nu(Z)\bigr)
=(\beta\gamma,\alpha\gamma,\alpha\beta).
\label{eq:single-qubit-event-factorization}
\end{align}
Conversely, every such triple is realized by independent $X$, $Y$, and $Z$ events with probabilities $(1-\alpha)/2$, $(1-\beta)/2$, and $(1-\gamma)/2$, respectively.
Thus Eq.~\eqref{eq:single-qubit-event-factorization}, with $\alpha,\beta,\gamma\in[-1,1]$, characterizes the single-qubit independent-event class exactly.
In particular, every distribution in this class satisfies
\begin{align}
\widehat\nu(X)\widehat\nu(Y)\widehat\nu(Z)
=(\alpha\beta\gamma)^2\geq0.
\label{eq:fourier-criterion}
\end{align}
The sign condition holds for arbitrary event probabilities.
More generally, complementing every event with $p_a>1/2$ writes $\nu$ as a deterministic Pauli translate of an independent-event distribution whose event probabilities are all at most $1/2$.
By Eq.~\eqref{eq:fourier-product}, the latter distribution has nonnegative Fourier coefficients.
However, even positivity of all Fourier coefficients does not guarantee an independent-event representation: the valid Fourier triple $(0.805,0.9,0.9)$ would require $\alpha^2=0.81/0.805>1$ in Eq.~\eqref{eq:single-qubit-event-factorization}.

For depolarizing noise with total error probability $p$, all three nontrivial Fourier coefficients equal $1-4p/3$.
For $p\leq3/4$, Eq.~\eqref{eq:single-qubit-event-factorization} is realized by $\alpha=\beta=\gamma=\sqrt{1-4p/3}$, corresponding to three independent flips with equal probability
\begin{align}
r=\frac{1-\sqrt{1-4p/3}}{2}.
\end{align}
For $p>3/4$, all three coefficients are negative and violate Eq.~\eqref{eq:fourier-criterion}.
Thus depolarizing noise belongs to the independent-event class exactly up to complete depolarization.

The exclusion test also applies to many-qubit channels through their single-qubit marginals.
Restricting each Pauli event to one qubit, while retaining its independent Bernoulli variable, produces an independent-event representation of that marginal.
Consequently, a negative Fourier product in any single-qubit marginal excludes an independent-event representation of the full channel, even one using multi-qubit events.

\subsection{Failure of the von Neumann--collision step}

Consider the $[[2,1,1]]$ code with stabilizer $X_1X_2$ and logical operators $\bar X=X_1$ and $\bar Z=Z_1Z_2$.
Let qubit~2 be noiseless and let qubit~1 undergo
\begin{align}
\nu_A=\frac1{16}(1,3,6,6).
\end{align}
The syndrome is trivial for $I_1$ and $X_1$ and nontrivial for $Y_1$ and $Z_1$.
Taking $Z_1$ as the reference error for the nontrivial syndrome, $Y_1\propto Z_1X_1$ has logical class $\bar X$.
The joint syndrome--logical distribution is therefore
\begin{align}
P(s,\ell)=\frac1{16}\begin{pmatrix}1&3&0&0\\6&6&0&0\end{pmatrix},
\label{eq:s-endpoint-counterexample-law}
\end{align}
with rows $s\in\{0,1\}$ and columns $\ell\in\{I,\bar X,\bar Y,\bar Z\}$.
Writing $h_{\mathrm{bin}}(x):=-x\log x-(1-x)\log(1-x)$ gives
\begin{align}
\Delta_1
&=\frac14h_{\mathrm{bin}}\left(\frac14\right)+\frac34\log2
\approx0.660444,\nonumber\\
\Delta_2
&=\log\frac{80}{41}\approx0.668455.
\end{align}
Hence $\Delta_1<\Delta_2$, or equivalently $I_c^{(1)}>I_c^{(2)}$.
Numerical evaluation also gives $\Delta_1<\Delta_2<\cdots<\Delta_{100}$ (verified using 100-th digits).
The Fourier coefficients of $\nu_A$ are $(-1/2,-1/8,-1/8)$ in the order $(X,Y,Z)$.
Their negative product excludes an independent-event representation by Eq.~\eqref{eq:fourier-criterion}.

The reversal is driven by syndrome reweighting.
With $p_s:=\sum_\ell P(s,\ell)$, $q_s(\ell):=P(s,\ell)/p_s$, and $Q_2(s):=p_s^2/\sum_t p_t^2$, the collision deficit satisfies
\begin{align}
e^{-\Delta_2}=\sum_sQ_2(s)\sum_\ell q_s(\ell)^2.
\end{align}
Keeping the original syndrome weights would give $-\log\sum_s p_s\sum_\ell q_s(\ell)^2\leq\Delta_1$ by Jensen's inequality and $H_2\leq H$.
Here, however, syndrome $s=1$ has probability $p_1=3/4$ but reweighted probability $Q_2(1)=9/10$.
Its logical posterior is uniform on two classes, whereas that of $s=0$ has probabilities $(1/4,3/4)$ on its support.
The shift toward the more uncertain syndrome produces the reversal.
The random-rank argument of Step~I excludes this net reversal for independent Bernoulli faults; it does not require more probable syndromes to have smaller conditional uncertainty.

\subsection{Failure of a higher-order step}

The next example detects every nonidentity single-qubit Pauli error, showing that a violation does not require distance one.
Consider the $[[4,1,2]]$ CSS code with stabilizer generators $X_1X_2X_3X_4$, $Z_1Z_2Z_3Z_4$, and $X_1X_2$, and logical operators $\bar X=X_1X_3$ and $\bar Z=Z_1Z_2$.
Let the noise be $\nu_1\otimes\nu_2\otimes\nu_3\otimes\nu_4$, where
\begin{align}
\nu_1&=\Bigl(\frac{49}{50},\frac1{50},0,0\Bigr), \qquad 
\nu_2=\Bigl(\frac{49}{50},0,\frac1{50},0\Bigr), \qquad 
\nu_3=\Bigl(\frac12,0,0,\frac12\Bigr), \qquad 
\nu_4=\Bigl(\frac25,\frac3{10},0,\frac3{10}\Bigr).
\end{align}
The first three factors are single Bernoulli $X$, $Y$, and $Z$ events, respectively.
The fourth has Fourier coefficients $(2/5,-1/5,2/5)$, whose product is negative.
Thus qubit~4 is the only factor outside the independent-event class, and the marginal argument excludes an independent-event representation of the full channel.
Summing error probabilities within each stabilizer coset gives
\begin{align}
e^{-\Delta_2}&=\frac{20417}{35858}, \qquad
e^{-2\Delta_3}=\frac{138827}{450740}.
\end{align}
The deficits at the first three orders are
\begin{align}
\Delta_1&\approx0.569448, \qquad 
\Delta_2\approx0.563199, \qquad
\Delta_3\approx0.588831.
\end{align}
Thus the von Neumann--collision inequality holds, whereas $\Delta_2<\Delta_3$ gives $I_c^{(2)}>I_c^{(3)}$.
By continuity, the violation persists under sufficiently small changes of $\nu_4$ that remain valid probability distributions, with the other factors fixed.

\vspace{5pt} \noindent {\bf Remarks.} Both counterexamples are inhomogeneous, with only one single-qubit factor outside the independent-event class.
For the $[[4,1,2]]$ code, homogeneous noise $\nu_A^{\otimes4}$ or $\nu_4^{\otimes4}$ satisfies $\Delta_1\geq\Delta_2\geq\cdots\geq\Delta_{100}$ in our numerical checks.
Whether homogeneous product Pauli noise can violate the hierarchy for some stabilizer code remains open.


\section{Matched-syndrome operation for general noise}
\label{app:balance}

\subsection{Syndrome instrument and syndrome-resolved states}
\label{app:general-instrument}

Let $V:\mathcal H_{\mathrm L}\to\mathcal H_Q$ be the encoding isometry, $\{\Pi_s\}$ the syndrome projectors, and $\cN$ an arbitrary quantum channel.
Each syndrome sector $\mathcal H_s:=\Pi_s\mathcal H_Q$ has dimension $\dL$; let $I_s$ denote its identity.
The branch maps $\cM_s(\rho_{\mathrm L}):=\Pi_s\cN(V\rho_{\mathrm L}V^\dagger)\Pi_s$ form a quantum instrument: each is completely positive, and their sum is trace preserving.
Unlike the Pauli case, the syndrome probabilities can depend on the logical input.
We define $p_s:=\Tr\cM_s(\piL)$ using the maximally mixed input $\piL:=I_{\mathrm L}/\dL$.
Since $\piL$ has full rank, $p_s=0$ implies $\cM_s=0$, so such branches can be omitted throughout.

Let $\Phi_{R\mathrm L}$ denote the density operator of a normalized maximally entangled state between the logical input and a reference $R$.
The normalized branch states and their output marginals are
\begin{align}
J_s&:=\frac{(\id_R\otimes\cM_s)(\Phi_{R\mathrm L})}{p_s},\nonumber\\
\omega_s&:=\Tr_RJ_s=\frac{\cM_s(\piL)}{p_s}.
\label{eq:s-branch-state}
\end{align}
Retaining the syndrome in a classical register gives $\Omega_{RSQ}:=\sum_sp_s\ketbra{s}_S\otimes J_s$ and $\Omega_{SQ}:=\Tr_R\Omega_{RSQ}$.
For general noise, we define $I_c^{(n)}:=H_n(\Omega_{SQ})-H_n(\Omega_{RSQ})$ and $\Delta_n:=\log\dL-I_c^{(n)}$ from these syndrome-resolved states.
These quantities need not agree with those before syndrome measurement, which can remove coherences between sectors.

For an integer $n\geq2$, matching the data syndrome to those of $n-1$ independent auxiliary inputs in $\piL$ gives the same operation as Eq.~\eqref{eq:postselection-operation}:
\begin{align}
\cP_n(\rho_{\mathrm L})
=\sum_sp_s^{n-1}\ketbra{s}_S\otimes\cM_s(\rho_{\mathrm L}).
\end{align}
Its acceptance probability at input $\piL$ is $Z_n:=\Tr\cP_n(\piL)=\sum_sp_s^n$.
Writing $Q_n(s):=p_s^n/Z_n$, the normalized postselected state obtained from $\Phi_{R\mathrm L}$ is
\begin{align}
\Theta_n&:=\frac{(\id_R\otimes\cP_n)(\Phi_{R\mathrm L})}{Z_n} =\sum_sQ_n(s)\ketbra{s}_S\otimes J_s.
\end{align}
This state is well defined without assuming input-independent success.
The classical syndrome flag gives $\Tr\Omega_{RSQ}^n=\sum_sp_s^n\Tr J_s^n$ and $\Tr\Omega_{SQ}^n=\sum_sp_s^n\Tr\omega_s^n$, hence
\begin{align}
e^{-(n-1)\Delta_n}
=\frac{\sum_sQ_n(s)\Tr J_s^n}
{\sum_sQ_n(s)\dL^{n-1}\Tr\omega_s^n}.
\label{eq:s-moment-factorization}
\end{align}

Define the optimal branch recovery fidelity and its $Q_n$-average by
\begin{align}
F(J_s)&:=\max_{\cR_s}\Tr\left[\Phi_{R\mathrm L}(\id_R\otimes\cR_s)(J_s)\right],\nonumber\\
F_n&:=\sum_sQ_n(s)F(J_s),
\label{eq:s-Fn}
\end{align}
where $\cR_s$ ranges over channels from $\mathcal H_s$ to $\mathcal H_{\mathrm L}$.
Because the syndrome is retained, the recoveries can be optimized separately for each branch.
Thus $F_n$ is the optimal recovery fidelity of $\Theta_n$; under $n$-balance, it is also the optimal recovered entanglement fidelity of the matched channel at input $\piL$.
For stochastic Pauli noise, $J_s=\sum_\ell q_s(\ell)\ketbra{\Phi_{s,\ell}}$ and $\omega_s=\pi_s:=I_s/\dL$, where the Bell states are embedded in sector $s$ as in the main text.
Then $F(J_s)=r_s:=\max_\ell q_s(\ell)$ and $F_n=1-\epsilon_n$.

\subsection{Balance}

As in the main text, the construction is \emph{$n$-balanced} if $\Tr\cP_n(\rho_{\mathrm L})=Z_n$ for every normalized logical input.
Write $E_s:=\cM_s^\dagger(I_s)$ for the branch effect, so $\Tr\cM_s(\rho_{\mathrm L})=\Tr(\rho_{\mathrm L}E_s)$ and $\sum_sE_s=I_{\mathrm L}$.
The effect governing acceptance is
\begin{align}
G_n:=\cP_n^\dagger(I_{SQ})=\sum_sp_s^{n-1}E_s.
\label{eq:s-effect}
\end{align}
Since $\Tr\cP_n(\rho_{\mathrm L})=\Tr(\rho_{\mathrm L}G_n)$, balance is equivalent to
\begin{align}
G_n=Z_nI_{\mathrm L}.
\label{eq:s-n-balance}
\end{align}
This is exactly the trace-preservation condition for $\widetilde{\cP}_n:=\cP_n/Z_n$, whose Choi state is then $\Theta_n$.

Without balance, normalizing the flagged output separately for each input does not define an affine map.
To prove this, write
\begin{align}
f_n(\rho)&:=\Tr\cP_n(\rho), \qquad 
\mathfrak N_n(\rho):=\frac{\cP_n(\rho)}{f_n(\rho)}.
\end{align}
Every normalized input has positive success probability, since the nonzero branches form a trace-preserving instrument and all their weights $p_s^{n-1}$ are positive.
Evaluating $\mathfrak N_n(\lambda\rho+(1-\lambda)\sigma)$ for $0<\lambda<1$ using linearity of $\cP_n$ shows that affinity would require
\begin{align}
[f_n(\rho)-f_n(\sigma)]
[\mathfrak N_n(\rho)-\mathfrak N_n(\sigma)]=0
\end{align}
for every pair of inputs, after dividing by $\lambda(1-\lambda)>0$.
If $f_n$ were nonconstant, choose $\rho,\sigma$ with different success probabilities.
Their normalized outputs must agree, and comparison of any third input with one of this pair forces the same output.
Thus $\mathfrak N_n$ would be constant.
Write this constant flagged state as $\sum_s\ketbra{s}_S\otimes\tau_s$.
Equality of the individual flag blocks gives
\begin{align}
p_s^{n-1}\cM_s(\rho)=f_n(\rho)\tau_s.
\end{align}
Dividing by $p_s^{n-1}$, taking traces, and summing over $s$ yields
\begin{align}
1=\sum_s\Tr\cM_s(\rho)
=f_n(\rho)\sum_sp_s^{1-n}\Tr\tau_s.
\end{align}
The last sum is input independent, forcing $f_n$ to be constant, a contradiction.
Hence the normalized flagged operation is affine if and only if $n$-balance holds.

A sufficient condition is \emph{branchwise trace scaling},
\begin{align}
E_s=\cM_s^\dagger(I_s)=p_sI_{\mathrm L}.
\label{eq:s-trace-scaling}
\end{align}
Then each $\cM_s/p_s$ is a channel and $G_n=\sum_sp_s^nI_{\mathrm L}=Z_nI_{\mathrm L}$ at every order.
Stochastic Pauli noise satisfies this condition because each Pauli error has a syndrome independent of the encoded logical state.
Balance therefore holds automatically for every integer $n\geq2$, without an independent-event assumption.

Trace scaling is not necessary.
For example, assume $\dL>1$ and at least $\dL$ syndrome sectors are available.
Measure an orthonormal logical basis $\{\ket{s}_{\mathrm L}\}$ and prepare a normalized state $\ket{\psi_s}\in\mathcal H_s$ in a distinct sector for each outcome:
\begin{align}
\cM_s(\rho)=\langle s|\rho|s\rangle\ketbra{\psi_s}.
\end{align}
These branches have $E_s=\ketbra{s}_{\mathrm L}$ and $p_s=1/\dL$.
Thus $G_n=\dL^{1-n}I_{\mathrm L}=Z_nI_{\mathrm L}$ at every order, although the individual effects are not proportional to the identity.

There is also an exact characterization of balance at every order.
Let $v_1,\ldots,v_{N_v}$ be the distinct nonzero values among the $p_s$, and define $\mathcal S_j:=\{s:p_s=v_j\}$.
Extend $G_n$ and $Z_n$ to $n=1$ by the same formulas, so $G_1=I_{\mathrm L}$ and $Z_1=1$.
Then
\begin{align}
G_n-Z_nI_{\mathrm L}
=\sum_{j=1}^{N_v}v_j^{n-1}
\left[\sum_{s\in\mathcal S_j}E_s-|\mathcal S_j|v_jI_{\mathrm L}\right].
\end{align}
For $n=1,\ldots,N_v$, the coefficients form an invertible Vandermonde matrix because the $v_j$ are distinct.
Since the $n=1$ equation holds automatically, balance at $n=2,\ldots,N_v$ is equivalent to
\begin{align}
\sum_{s\in\mathcal S_j}E_s=|\mathcal S_j|v_jI_{\mathrm L}
\label{eq:s-class-balance}
\end{align}
for every $j$, which in turn implies balance at every order.
Equation~\eqref{eq:s-class-balance} says that each group of syndromes with equal $p_s$ has an input-independent total probability.
It reduces to branchwise trace scaling when all nonzero $p_s$ are distinct and follows from trace preservation when they are all equal.

An explicit failure occurs for amplitude damping on the first qubit of the repetition code $\ket{0_{\mathrm L}}=\ket{00}$, $\ket{1_{\mathrm L}}=\ket{11}$~\cite{FletcherShorWin2008}.
For damping strength $\gamma\in[0,1]$, write $u:=\langle1_{\mathrm L}|\rho_{\mathrm L}|1_{\mathrm L}\rangle$.
The odd-parity syndrome has probability $\Pr(\mathrm{odd}\mid\rho_{\mathrm L})=\gamma u$, while the auxiliary syndrome probabilities are $p_{\mathrm{odd}}=\gamma/2$ and $p_{\mathrm{even}}=1-\gamma/2$.
Consequently,
\begin{align}
f_n(\rho_{\mathrm L})
=\left(1-\frac{\gamma}{2}\right)^{n-1}(1-\gamma u)
+\left(\frac{\gamma}{2}\right)^{n-1}\gamma u.
\end{align}
For $0<\gamma<1$ and $n\geq2$, the two weights differ, so acceptance depends on $u$ and balance fails.
At $\gamma=0$ or $1$, balance holds; at $\gamma=1$ the auxiliary weights are equal despite input-dependent branch probabilities.

\subsection{Approximate balance from recovery}

High recovery fidelity constrains the input dependence of the syndrome probabilities.
The reference marginal of each normalized branch state is
\begin{align}
J_{s,R}=\frac{E_s^{\mathsf T}}{\dL p_s},
\label{eq:s-reference-effect}
\end{align}
where the transpose is taken in the Schmidt basis defining $\Phi_{R\mathrm L}$.
Recovery acts only on the output and therefore leaves this marginal unchanged.
Applying trace-distance contractivity under partial trace and the Fuchs--van de Graaf inequalities~\cite{FuchsVanDeGraaf1999} to an optimally recovered branch gives
\begin{align}
\left\|\frac{E_s}{p_s}-I_{\mathrm L}\right\|_1
&=\dL\left\|J_{s,R}-\frac{I_R}{\dL}\right\|_1 \leq2\dL\sqrt{1-F(J_s)}.
\end{align}
Here the equality uses invariance of the trace norm under transposition.
Averaging over $Q_n$ and applying Jensen's inequality yields
\begin{align}
\sum_sQ_n(s)\left\|\frac{E_s}{p_s}-I_{\mathrm L}\right\|_1
&\leq2\dL\sum_sQ_n(s)\sqrt{1-F(J_s)} \leq2\dL\sqrt{1-F_n}.
\label{eq:s-average-branch-balance}
\end{align}
Since
\begin{align}
\frac{G_n}{Z_n}-I_{\mathrm L}
=\sum_sQ_n(s)\left[\frac{E_s}{p_s}-I_{\mathrm L}\right],
\end{align}
the triangle inequality and $\|A\|_\infty\leq\|A\|_1$ give
\begin{align}
\sup_{\rho_{\mathrm L}}
\left|\frac{\Tr\cP_n(\rho_{\mathrm L})}{Z_n}-1\right|
&=\left\|\frac{G_n}{Z_n}-I_{\mathrm L}\right\|_\infty \leq2\dL\sqrt{1-F_n},
\label{eq:s-average-leak}
\end{align}
where the supremum is over normalized logical inputs.
At fixed logical dimension, Eq.~\eqref{eq:s-average-leak} shows that $F_n\to1$ makes the relative input dependence of the matching probability vanish uniformly.

\section{Recovery bounds and dimension dependence}
\label{app:recovery-bounds}

\subsection{Finite-size recovery bounds}

We use the normalized branch states $J_s$, output marginals $\omega_s:=\Tr_RJ_s$, and recovery fidelities $F(J_s)$ and $F_n$ defined in Appendix~\ref{app:balance}.
Each syndrome sector has dimension $\dL$, and $\pi_s:=I_s/\dL$ denotes its maximally mixed state.
Define
\begin{align}
A_n&:=\sum_sQ_n(s)\Tr J_s^n,\nonumber\\
U_n&:=\sum_sQ_n(s)\dL^{n-1}\Tr\omega_s^n,\nonumber\\
B_n&:=e^{-(n-1)\Delta_n}=\frac{A_n}{U_n}.
\label{eq:s-moment-decomposition}
\end{align}
The last identity follows from Eq.~\eqref{eq:s-moment-factorization}.
Since each $J_s$ is normalized and each $\omega_s$ is supported on a $\dL$-dimensional sector, $0<A_n\leq1$ and $U_n\geq1$.
Thus $0<B_n\leq A_n\leq1$.
The following bounds require no balance assumption.

\begin{proposition}[Finite-size recovery bounds]
\label{prop:recovery}
For every integer $n\geq2$,
\begin{align}
F_n&\geq B_n^2,\nonumber\\
B_n&\geq\frac{F_n^n}{1+n(n-1)\dL^{n-1}(1-F_n)}.
\label{eq:s-recovery-bounds}
\end{align}
At fixed $n$ and $\dL$, $\Delta_n\to0$ if and only if $F_n\to1$.
At fixed $n$, the implication from $\Delta_n\to0$ to $F_n\to1$ is uniform in $\dL$.
\end{proposition}

\begin{proof}
Write $f_s:=F(J_s)$.
For a recovery channel $\cR_s:\mathcal H_s\to\mathcal H_{\mathrm L}$, the operator $X_s:=(\id_R\otimes\cR_s^\dagger)(\Phi_{R\mathrm L})$ is positive and satisfies $\Tr_RX_s=\pi_s$ because $\cR_s^\dagger$ is unital.
Conversely, channel--state duality identifies every operator satisfying these constraints with the adjoint of a recovery channel~\cite{Choi1975,Jamiolkowski1972}.
Moving the recovery map to the other factor in the fidelity overlap therefore gives
\begin{align}
f_s=\max\left\{\Tr(J_sX_s):X_s\geq0,\Tr_RX_s=\pi_s\right\}.
\label{eq:s-SDP}
\end{align}
Every feasible $X_s$ has unit trace and is therefore a normalized density operator.

\emph{Forward bound.--}
Let $P_s$ project onto the support of $\omega_s$.
All powers of $\omega_s$ act on the syndrome factor, and inverse powers are taken on its support and set to zero on its kernel.
Consider
\begin{align}
\widehat X_s
&:=\frac{1}{\dL}\omega_s^{-1/2}J_s\omega_s^{-1/2}
+\frac{I_R\otimes(I_s-P_s)}{\dL^2}.
\label{eq:s-recovery-candidate}
\end{align}
This operator is positive, and $\Tr_R\widehat X_s=P_s/\dL+(I_s-P_s)/\dL=\pi_s$, so it is feasible in Eq.~\eqref{eq:s-SDP}.
Moreover, $J_s$ is supported on $\mathcal H_R\otimes\operatorname{supp}\omega_s$, so the second term has zero overlap with $J_s$.
Hilbert--Schmidt Cauchy--Schwarz gives
\begin{align}
(\Tr J_s^2)^2
&\leq\Tr\left[(\omega_s^{-1/4}J_s\omega_s^{-1/4})^2\right]
\Tr\left[(\omega_s^{1/4}J_s\omega_s^{1/4})^2\right]\nonumber\\
&=\dL\Tr(J_s\widehat X_s)
\Tr(J_s\omega_s^{1/2}J_s\omega_s^{1/2}) \leq\dL f_s\Tr\omega_s^2.
\label{eq:s-recovery-purity}
\end{align}
The first inequality uses that the Hilbert--Schmidt inner product of the two sandwiched operators is $\Tr J_s^2$.
For the last inequality, feasibility gives $\Tr(J_s\widehat X_s)\leq f_s$, while $J_s\leq I_R\otimes I_s$ gives $\omega_s^{1/2}J_s\omega_s^{1/2}\leq\omega_s$.
Tracing the latter inequality against $J_s\geq0$ bounds the second factor by $\Tr(J_s\omega_s)=\Tr\omega_s^2$.

For $n\geq2$, normalization gives $\Tr J_s^n\leq\Tr J_s^2$.
If $\lambda_i$ are the eigenvalues of $\omega_s$, Jensen's inequality with weights $\lambda_i$ gives $(\sum_i\lambda_i^2)^{n-1}\leq\sum_i\lambda_i^n$.
Consequently,
\begin{align}
\dL\Tr\omega_s^2
&\leq\left(\dL^{n-1}\Tr\omega_s^n\right)^{1/(n-1)} \leq\dL^{n-1}\Tr\omega_s^n,
\end{align}
where the second inequality uses $\dL^{n-1}\Tr\omega_s^n\geq1$.
Equation~\eqref{eq:s-recovery-purity} therefore implies $(\Tr J_s^n)^2\leq f_s\dL^{n-1}\Tr\omega_s^n$.
Averaging over $Q_n$ and applying Cauchy--Schwarz yields
\begin{align}
A_n^2
&\leq\left(\sum_sQ_n(s)\sqrt{f_s\dL^{n-1}\Tr\omega_s^n}\right)^2 \leq F_nU_n,\nonumber\\
F_n&\geq\frac{A_n^2}{U_n}=A_nB_n\geq B_n^2.
\label{eq:s-master-forward}
\end{align}
The last step uses $A_n\geq B_n$.

\emph{Converse bound.--}
Let $X_s$ be an optimizer in Eq.~\eqref{eq:s-SDP}.
For the squared Uhlmann fidelity $\mathcal F(\rho,\sigma):=\|\sqrt\rho\sqrt\sigma\|_1^2$~\cite{Uhlmann1976}, the inequality $\|A\|_1\geq\|A\|_2$ gives $\mathcal F(J_s,X_s)\geq\Tr(J_sX_s)=f_s$.
The Fuchs--van de Graaf inequalities~\cite{FuchsVanDeGraaf1999} and contractivity of trace distance under partial trace then give
\begin{align}
\|\omega_s-\pi_s\|_1
&\leq\|J_s-X_s\|_1 \leq2\sqrt{1-f_s}.
\label{eq:s-tau-F}
\end{align}

The second derivative of $x^n$ is at most $n(n-1)$ on $[0,1]$.
Apply Taylor's theorem about $1/\dL$ to each eigenvalue of $\omega_s$.
The linear terms sum to zero because $\Tr\omega_s=1$, so
\begin{align}
\Tr\omega_s^n-\dL^{1-n}
&\leq\frac{n(n-1)}{2}\|\omega_s-\pi_s\|_2^2 \leq\frac{n(n-1)}{4}\|\omega_s-\pi_s\|_1^2 \leq n(n-1)(1-f_s).
\label{eq:s-power-quadratic}
\end{align}
Here $\|\cdot\|_2$ is the Hilbert--Schmidt norm.
For the second inequality, the positive and negative parts of a traceless Hermitian operator $D$ have equal trace, which gives $\|D\|_2^2\leq\|D\|_1^2/2$.
The last inequality follows from Eq.~\eqref{eq:s-tau-F}.
Multiplying by $\dL^{n-1}$ and averaging over $Q_n$ gives
\begin{align}
U_n\leq1+n(n-1)\dL^{n-1}(1-F_n).
\label{eq:s-output-moment-converse}
\end{align}
Every feasible $X_s$ is a density operator, so $f_s\leq\lambda_{\max}(J_s)$ and hence $\Tr J_s^n\geq f_s^n$.
Jensen's inequality gives $A_n\geq\sum_sQ_n(s)f_s^n\geq F_n^n$.
Combining the numerator and denominator bounds proves
\begin{align}
B_n\geq\frac{F_n^n}{1+n(n-1)\dL^{n-1}(1-F_n)}.
\label{eq:s-master-converse}
\end{align}
\end{proof}

In particular, the forward bound gives a dimension-independent estimate linear in the deficit:
\begin{align}
1-F_n
&\leq1-e^{-2(n-1)\Delta_n} \leq2(n-1)\Delta_n.
\label{eq:s-deficit-error}
\end{align}

\subsection{Dependence on the logical dimension}

\emph{Maximally mixed output marginals.--}
Suppose $\omega_s=\pi_s$ for every branch with $p_s>0$, equivalently $\cM_s(I_{\mathrm L})=p_sI_s$.
Then $U_n=1$, and $X_s=J_s$ is feasible in Eq.~\eqref{eq:s-SDP}.
Thus $\Tr J_s^n\leq\Tr J_s^2\leq f_s$.
Averaging and combining with $A_n\geq F_n^n$ yields
\begin{align}
F_n^n\leq B_n\leq F_n.
\label{eq:s-unital-delta}
\end{align}
At fixed $n$, $\Delta_n\to0$ is therefore equivalent to $F_n\to1$, even when $\dL$ grows.
No condition on the reference marginals is needed.
Branchwise trace scaling, $\cM_s^\dagger(I_s)=p_sI_{\mathrm L}$, is a separate sufficient condition for balance, as discussed in Appendix~\ref{app:balance}.
Stochastic Pauli branches satisfy both conditions and have $F_n=1-\epsilon_n$, recovering the main-text bound without the independent-event assumption.

\emph{Nonunital obstruction.--}
Fix an integer $n\geq2$ and consider two branches with outputs in distinct syndrome sectors.
Identifying each sector with the logical space, take an identity branch $\cM_0(\rho)=p_0\rho$ and a replacer branch $\cM_\star(\rho)=p_\star\ketbra{0}\Tr\rho$, where $p_0+p_\star=1$.
Both branches are trace scaling, so balance holds at every order.
To give the replacer a prescribed matched-syndrome weight $Q_n(\star)=q\in(0,1)$, choose $p_\star/p_0=[q/(1-q)]^{1/n}$.
The normalized branch states are $J_0=\Phi_{R\mathrm L}$ and $J_\star=(I_R/\dL)\otimes\ketbra{0}$.
The identity branch has recovery fidelity one.
Any recovery of the replacer branch leaves a product state $(I_R/\dL)\otimes\tau_{\mathrm L}$, whose overlap with $\Phi_{R\mathrm L}$ is $\dL^{-2}$.
Consequently,
\begin{align}
F_n&=1-q+q\dL^{-2}, \qquad A_n=1-q+q\dL^{1-n}, \qquad U_n=1-q+q\dL^{n-1}.
\label{eq:s-rare-replacer-exact}
\end{align}
The replacer contributes at most $q$ to the recovery error but $q\dL^{n-1}$ to the output moment $U_n$.
Thus a branch of vanishing weight can still dominate the denominator of $B_n$.

Taking $q=1/\log\dL$ for sufficiently large $\dL$ gives $F_n\to1$, $A_n\to1$, and $U_n\sim\dL^{n-1}/\log\dL$.
Hence
\begin{align}
\Delta_n
=\log\dL-\frac{\log\log\dL}{n-1}+o(1).
\label{eq:s-rare-replacer-summary}
\end{align}
High recovery fidelity therefore does not imply a vanishing deficit uniformly in dimension, even under exact balance.

The same example determines the necessary dimension scaling for a converse of the form in Eq.~\eqref{eq:s-master-converse}.
Taking $q=c\dL^{1-n}$ with fixed $c>0$ gives
\begin{align}
1-F_n&\sim c\dL^{1-n}, \qquad B_n\longrightarrow\frac{1}{1+c}.
\end{align}
Suppose a bound $B_n\geq F_n^n/[1+C_ng(\dL)(1-F_n)]$ held uniformly in $\dL$, with $C_n$ independent of $\dL$ and $g(\dL)=o(\dL^{n-1})$.
Its right-hand side would tend to one, contradicting the displayed limit of $B_n$.
Therefore, the dimension scaling in Eq.~\eqref{eq:s-master-converse} is sharp within this class of bounds, up to its dimension-independent coefficient.

\vspace{5pt} \noindent {\bf Remark.}  
The bounds concern the normalized postselected state $\Theta_n$, with recovery optimized over all syndrome-conditioned channels.
Under $n$-balance, $F_n$ is also the optimal recovered entanglement fidelity of $\widetilde{\cP}_n$ at input $\piL$.
For stochastic Pauli noise, syndrome measurement leaves the state unchanged and the deficit agrees with the unmeasured replica quantity; for general noise, it can remove coherences between sectors.
Efficient decoding requires a further analysis of the recovery implementation.

The distinction between absolute and rate-normalized saturation is illustrated by a channel that acts identically on $k-1$ logical qubits and fully dephases the last.
With a single syndrome,
\begin{align}
I_c^{(n)}&=(k-1)\log2,\qquad \Delta_n=\log2, \qquad \frac{I_c^{(n)}}{k\log2}\longrightarrow1,
\end{align}
while $F_n=1/2$ for every $k$.


\begin{thebibliography}{50}%
\makeatletter
\providecommand \@ifxundefined [1]{%
 \@ifx{#1\undefined}
}%
\providecommand \@ifnum [1]{%
 \ifnum #1\expandafter \@firstoftwo
 \else \expandafter \@secondoftwo
 \fi
}%
\providecommand \@ifx [1]{%
 \ifx #1\expandafter \@firstoftwo
 \else \expandafter \@secondoftwo
 \fi
}%
\providecommand \natexlab [1]{#1}%
\providecommand \enquote  [1]{``#1''}%
\providecommand \bibnamefont  [1]{#1}%
\providecommand \bibfnamefont [1]{#1}%
\providecommand \citenamefont [1]{#1}%
\providecommand \href@noop [0]{\@secondoftwo}%
\providecommand \href [0]{\begingroup \@sanitize@url \@href}%
\providecommand \@href[1]{\@@startlink{#1}\@@href}%
\providecommand \@@href[1]{\endgroup#1\@@endlink}%
\providecommand \@sanitize@url [0]{\catcode `\\12\catcode `\$12\catcode `\&12\catcode `\#12\catcode `\^12\catcode `\_12\catcode `\%12\relax}%
\providecommand \@@startlink[1]{}%
\providecommand \@@endlink[0]{}%
\providecommand \url  [0]{\begingroup\@sanitize@url \@url }%
\providecommand \@url [1]{\endgroup\@href {#1}{\urlprefix }}%
\providecommand \urlprefix  [0]{URL }%
\providecommand \Eprint [0]{\href }%
\providecommand \doibase [0]{https://doi.org/}%
\providecommand \selectlanguage [0]{\@gobble}%
\providecommand \bibinfo  [0]{\@secondoftwo}%
\providecommand \bibfield  [0]{\@secondoftwo}%
\providecommand \translation [1]{[#1]}%
\providecommand \BibitemOpen [0]{}%
\providecommand \bibitemStop [0]{}%
\providecommand \bibitemNoStop [0]{.\EOS\space}%
\providecommand \EOS [0]{\spacefactor3000\relax}%
\providecommand \BibitemShut  [1]{\csname bibitem#1\endcsname}%
\let\auto@bib@innerbib\@empty
\bibitem [{\citenamefont {Lee}\ \emph {et~al.}(2025)\citenamefont {Lee}, \citenamefont {You},\ and\ \citenamefont {Xu}}]{lee2022symmetry}%
  \BibitemOpen
  \bibfield  {author} {\bibinfo {author} {\bibfnamefont {J.~Y.}\ \bibnamefont {Lee}}, \bibinfo {author} {\bibfnamefont {Y.-Z.}\ \bibnamefont {You}},\ and\ \bibinfo {author} {\bibfnamefont {C.}~\bibnamefont {Xu}},\ }\bibfield  {title} {\bibinfo {title} {Symmetry protected topological phases under decoherence},\ }\href {https://doi.org/10.22331/q-2025-01-23-1607} {\bibfield  {journal} {\bibinfo  {journal} {{Quantum}}\ }\textbf {\bibinfo {volume} {9}},\ \bibinfo {pages} {1607} (\bibinfo {year} {2025})}\BibitemShut {NoStop}%
\bibitem [{\citenamefont {Lee}\ \emph {et~al.}(2023)\citenamefont {Lee}, \citenamefont {Jian},\ and\ \citenamefont {Xu}}]{Lee_2023}%
  \BibitemOpen
  \bibfield  {author} {\bibinfo {author} {\bibfnamefont {J.~Y.}\ \bibnamefont {Lee}}, \bibinfo {author} {\bibfnamefont {C.-M.}\ \bibnamefont {Jian}},\ and\ \bibinfo {author} {\bibfnamefont {C.}~\bibnamefont {Xu}},\ }\bibfield  {title} {\bibinfo {title} {Quantum criticality under decoherence or weak measurement},\ }\bibfield  {journal} {\bibinfo  {journal} {PRX Quantum}\ }\textbf {\bibinfo {volume} {4}},\ \href {https://doi.org/10.1103/prxquantum.4.030317} {10.1103/prxquantum.4.030317} (\bibinfo {year} {2023})\BibitemShut {NoStop}%
\bibitem [{\citenamefont {Fan}\ \emph {et~al.}(2024)\citenamefont {Fan}, \citenamefont {Bao}, \citenamefont {Altman},\ and\ \citenamefont {Vishwanath}}]{Fan2024}%
  \BibitemOpen
  \bibfield  {author} {\bibinfo {author} {\bibfnamefont {R.}~\bibnamefont {Fan}}, \bibinfo {author} {\bibfnamefont {Y.}~\bibnamefont {Bao}}, \bibinfo {author} {\bibfnamefont {E.}~\bibnamefont {Altman}},\ and\ \bibinfo {author} {\bibfnamefont {A.}~\bibnamefont {Vishwanath}},\ }\bibfield  {title} {\bibinfo {title} {Diagnostics of mixed-state topological order and breakdown of quantum memory},\ }\href {https://doi.org/10.1103/PRXQuantum.5.020343} {\bibfield  {journal} {\bibinfo  {journal} {PRX Quantum}\ }\textbf {\bibinfo {volume} {5}},\ \bibinfo {pages} {020343} (\bibinfo {year} {2024})},\ \Eprint {https://arxiv.org/abs/2301.05689} {arXiv:2301.05689} \BibitemShut {NoStop}%
\bibitem [{\citenamefont {Lee}(2025)}]{LeeCI2024}%
  \BibitemOpen
  \bibfield  {author} {\bibinfo {author} {\bibfnamefont {J.~Y.}\ \bibnamefont {Lee}},\ }\bibfield  {title} {\bibinfo {title} {Exact calculations of coherent information for toric codes under decoherence: Identifying the fundamental error threshold},\ }\href {https://doi.org/10.1103/hlfh-86yz} {\bibfield  {journal} {\bibinfo  {journal} {Phys. Rev. Lett.}\ }\textbf {\bibinfo {volume} {134}},\ \bibinfo {pages} {250601} (\bibinfo {year} {2025})}\BibitemShut {NoStop}%
\bibitem [{\citenamefont {Niwa}\ and\ \citenamefont {Lee}(2025)}]{NiwaLee2025}%
  \BibitemOpen
  \bibfield  {author} {\bibinfo {author} {\bibfnamefont {R.}~\bibnamefont {Niwa}}\ and\ \bibinfo {author} {\bibfnamefont {J.~Y.}\ \bibnamefont {Lee}},\ }\bibfield  {title} {\bibinfo {title} {Coherent information for {Calderbank--Shor--Steane} codes under decoherence},\ }\href {https://doi.org/10.1103/PhysRevA.111.032402} {\bibfield  {journal} {\bibinfo  {journal} {Physical Review A}\ }\textbf {\bibinfo {volume} {111}},\ \bibinfo {pages} {032402} (\bibinfo {year} {2025})},\ \Eprint {https://arxiv.org/abs/2407.02564} {arXiv:2407.02564} \BibitemShut {NoStop}%
\bibitem [{\citenamefont {Kim}\ \emph {et~al.}(2026)\citenamefont {Kim}, \citenamefont {Altman},\ and\ \citenamefont {Lee}}]{kim2024errorthresholdsykcodes}%
  \BibitemOpen
  \bibfield  {author} {\bibinfo {author} {\bibfnamefont {J.}~\bibnamefont {Kim}}, \bibinfo {author} {\bibfnamefont {E.}~\bibnamefont {Altman}},\ and\ \bibinfo {author} {\bibfnamefont {J.~Y.}\ \bibnamefont {Lee}},\ }\bibfield  {title} {\bibinfo {title} {Error threshold of sachdev-ye-kitaev models from strong-to-weak parity symmetry breaking},\ }\href {https://doi.org/10.1103/srf2-1f6d} {\bibfield  {journal} {\bibinfo  {journal} {Phys. Rev. B}\ }\textbf {\bibinfo {volume} {114}},\ \bibinfo {pages} {L051101} (\bibinfo {year} {2026})}\BibitemShut {NoStop}%
\bibitem [{\citenamefont {Li}\ and\ \citenamefont {Mong}(2025)}]{LiMong2025}%
  \BibitemOpen
  \bibfield  {author} {\bibinfo {author} {\bibfnamefont {Z.}~\bibnamefont {Li}}\ and\ \bibinfo {author} {\bibfnamefont {R.~S.~K.}\ \bibnamefont {Mong}},\ }\bibfield  {title} {\bibinfo {title} {Replica topological order in quantum mixed states and quantum error correction},\ }\href {https://doi.org/10.1103/PhysRevB.111.125106} {\bibfield  {journal} {\bibinfo  {journal} {Physical Review B}\ }\textbf {\bibinfo {volume} {111}},\ \bibinfo {pages} {125106} (\bibinfo {year} {2025})},\ \Eprint {https://arxiv.org/abs/2402.09516} {arXiv:2402.09516} \BibitemShut {NoStop}%
\bibitem [{\citenamefont {Lyons}(2024)}]{Lyons2024}%
  \BibitemOpen
  \bibfield  {author} {\bibinfo {author} {\bibfnamefont {A.}~\bibnamefont {Lyons}},\ }\href@noop {} {\bibinfo {title} {Understanding stabilizer codes under local decoherence through a general statistical mechanics mapping}} (\bibinfo {year} {2024}),\ \Eprint {https://arxiv.org/abs/2403.03955} {arXiv:2403.03955 [quant-ph]} \BibitemShut {NoStop}%
\bibitem [{\citenamefont {Chen}\ and\ \citenamefont {Grover}(2024)}]{Chen2024PRL}%
  \BibitemOpen
  \bibfield  {author} {\bibinfo {author} {\bibfnamefont {Y.-H.}\ \bibnamefont {Chen}}\ and\ \bibinfo {author} {\bibfnamefont {T.}~\bibnamefont {Grover}},\ }\bibfield  {title} {\bibinfo {title} {Separability transitions in topological states induced by local decoherence},\ }\href {https://doi.org/10.1103/PhysRevLett.132.170602} {\bibfield  {journal} {\bibinfo  {journal} {Phys. Rev. Lett.}\ }\textbf {\bibinfo {volume} {132}},\ \bibinfo {pages} {170602} (\bibinfo {year} {2024})}\BibitemShut {NoStop}%
\bibitem [{\citenamefont {Wang}\ \emph {et~al.}(2025{\natexlab{a}})\citenamefont {Wang}, \citenamefont {Wu},\ and\ \citenamefont {Wang}}]{Wang2025}%
  \BibitemOpen
  \bibfield  {author} {\bibinfo {author} {\bibfnamefont {Z.}~\bibnamefont {Wang}}, \bibinfo {author} {\bibfnamefont {Z.}~\bibnamefont {Wu}},\ and\ \bibinfo {author} {\bibfnamefont {Z.}~\bibnamefont {Wang}},\ }\bibfield  {title} {\bibinfo {title} {Intrinsic mixed-state topological order},\ }\href {https://doi.org/10.1103/PRXQuantum.6.010314} {\bibfield  {journal} {\bibinfo  {journal} {PRX Quantum}\ }\textbf {\bibinfo {volume} {6}},\ \bibinfo {pages} {010314} (\bibinfo {year} {2025}{\natexlab{a}})}\BibitemShut {NoStop}%
\bibitem [{\citenamefont {Sohal}\ and\ \citenamefont {Prem}(2025)}]{Sohal2025}%
  \BibitemOpen
  \bibfield  {author} {\bibinfo {author} {\bibfnamefont {R.}~\bibnamefont {Sohal}}\ and\ \bibinfo {author} {\bibfnamefont {A.}~\bibnamefont {Prem}},\ }\bibfield  {title} {\bibinfo {title} {Noisy approach to intrinsically mixed-state topological order},\ }\href {https://doi.org/10.1103/PRXQuantum.6.010313} {\bibfield  {journal} {\bibinfo  {journal} {PRX Quantum}\ }\textbf {\bibinfo {volume} {6}},\ \bibinfo {pages} {010313} (\bibinfo {year} {2025})}\BibitemShut {NoStop}%
\bibitem [{\citenamefont {Yang}\ \emph {et~al.}(2025)\citenamefont {Yang}, \citenamefont {Shi},\ and\ \citenamefont {Lee}}]{yang2025topologicalmixedstatesphases}%
  \BibitemOpen
  \bibfield  {author} {\bibinfo {author} {\bibfnamefont {T.-H.}\ \bibnamefont {Yang}}, \bibinfo {author} {\bibfnamefont {B.}~\bibnamefont {Shi}},\ and\ \bibinfo {author} {\bibfnamefont {J.~Y.}\ \bibnamefont {Lee}},\ }\href {https://arxiv.org/abs/2506.04221} {\bibinfo {title} {Topological mixed states: Phases of matter from axiomatic approaches}} (\bibinfo {year} {2025}),\ \Eprint {https://arxiv.org/abs/2506.04221} {arXiv:2506.04221 [cond-mat.str-el]} \BibitemShut {NoStop}%
\bibitem [{\citenamefont {Su}\ \emph {et~al.}(2024)\citenamefont {Su}, \citenamefont {Yang},\ and\ \citenamefont {Jian}}]{SuYangJian2024}%
  \BibitemOpen
  \bibfield  {author} {\bibinfo {author} {\bibfnamefont {K.}~\bibnamefont {Su}}, \bibinfo {author} {\bibfnamefont {Z.}~\bibnamefont {Yang}},\ and\ \bibinfo {author} {\bibfnamefont {C.-M.}\ \bibnamefont {Jian}},\ }\bibfield  {title} {\bibinfo {title} {Tapestry of dualities in decohered quantum error correction codes},\ }\href {https://doi.org/10.1103/PhysRevB.110.085158} {\bibfield  {journal} {\bibinfo  {journal} {Physical Review B}\ }\textbf {\bibinfo {volume} {110}},\ \bibinfo {pages} {085158} (\bibinfo {year} {2024})},\ \Eprint {https://arxiv.org/abs/2401.17359} {arXiv:2401.17359} \BibitemShut {NoStop}%
\bibitem [{\citenamefont {{Vijay}}\ and\ \citenamefont {{Lee}}(2025{\natexlab{a}})}]{vijay2025informationcriticalphasesdecoherence}%
  \BibitemOpen
  \bibfield  {author} {\bibinfo {author} {\bibfnamefont {A.}~\bibnamefont {{Vijay}}}\ and\ \bibinfo {author} {\bibfnamefont {J.~Y.}\ \bibnamefont {{Lee}}},\ }\bibfield  {title} {\bibinfo {title} {{Information Critical Phases under Decoherence}},\ }\href {https://doi.org/10.48550/arXiv.2512.22121} {\bibfield  {journal} {\bibinfo  {journal} {arXiv e-prints}\ ,\ \bibinfo {eid} {arXiv:2512.22121}} (\bibinfo {year} {2025}{\natexlab{a}})},\ \Eprint {https://arxiv.org/abs/2512.22121} {arXiv:2512.22121 [quant-ph]} \BibitemShut {NoStop}%
\bibitem [{\citenamefont {{Vijay}}\ and\ \citenamefont {{Lee}}(2025{\natexlab{b}})}]{vijay2025holographicallyemergentgaugetheory}%
  \BibitemOpen
  \bibfield  {author} {\bibinfo {author} {\bibfnamefont {A.}~\bibnamefont {{Vijay}}}\ and\ \bibinfo {author} {\bibfnamefont {J.~Y.}\ \bibnamefont {{Lee}}},\ }\bibfield  {title} {\bibinfo {title} {{Holographically Emergent Gauge Theory in Symmetric Quantum Circuits}},\ }\href {https://doi.org/10.48550/arXiv.2511.21685} {\bibfield  {journal} {\bibinfo  {journal} {arXiv e-prints}\ ,\ \bibinfo {eid} {arXiv:2511.21685}} (\bibinfo {year} {2025}{\natexlab{b}})},\ \Eprint {https://arxiv.org/abs/2511.21685} {arXiv:2511.21685 [quant-ph]} \BibitemShut {NoStop}%
\bibitem [{\citenamefont {Colmenarez}\ \emph {et~al.}(2024)\citenamefont {Colmenarez}, \citenamefont {Huang}, \citenamefont {Diehl},\ and\ \citenamefont {M\"uller}}]{PhysRevResearch.6.L042014}%
  \BibitemOpen
  \bibfield  {author} {\bibinfo {author} {\bibfnamefont {L.}~\bibnamefont {Colmenarez}}, \bibinfo {author} {\bibfnamefont {Z.-M.}\ \bibnamefont {Huang}}, \bibinfo {author} {\bibfnamefont {S.}~\bibnamefont {Diehl}},\ and\ \bibinfo {author} {\bibfnamefont {M.}~\bibnamefont {M\"uller}},\ }\bibfield  {title} {\bibinfo {title} {Accurate optimal quantum error correction thresholds from coherent information},\ }\href {https://doi.org/10.1103/PhysRevResearch.6.L042014} {\bibfield  {journal} {\bibinfo  {journal} {Phys. Rev. Res.}\ }\textbf {\bibinfo {volume} {6}},\ \bibinfo {pages} {L042014} (\bibinfo {year} {2024})}\BibitemShut {NoStop}%
\bibitem [{\citenamefont {Temkin}\ \emph {et~al.}(2025)\citenamefont {Temkin}, \citenamefont {Weinstein}, \citenamefont {Fan}, \citenamefont {Podolsky},\ and\ \citenamefont {Altman}}]{temkin2025chargeinformedquantumerrorcorrection}%
  \BibitemOpen
  \bibfield  {author} {\bibinfo {author} {\bibfnamefont {V.}~\bibnamefont {Temkin}}, \bibinfo {author} {\bibfnamefont {Z.}~\bibnamefont {Weinstein}}, \bibinfo {author} {\bibfnamefont {R.}~\bibnamefont {Fan}}, \bibinfo {author} {\bibfnamefont {D.}~\bibnamefont {Podolsky}},\ and\ \bibinfo {author} {\bibfnamefont {E.}~\bibnamefont {Altman}},\ }\href {https://arxiv.org/abs/2512.22119} {\bibinfo {title} {Charge-informed quantum error correction}} (\bibinfo {year} {2025}),\ \Eprint {https://arxiv.org/abs/2512.22119} {arXiv:2512.22119 [quant-ph]} \BibitemShut {NoStop}%
\bibitem [{\citenamefont {Wang}\ \emph {et~al.}(2025{\natexlab{b}})\citenamefont {Wang}, \citenamefont {Fan}, \citenamefont {Wang}, \citenamefont {Garratt},\ and\ \citenamefont {Altman}}]{wang2025fractionalquantumhallstates}%
  \BibitemOpen
  \bibfield  {author} {\bibinfo {author} {\bibfnamefont {Z.}~\bibnamefont {Wang}}, \bibinfo {author} {\bibfnamefont {R.}~\bibnamefont {Fan}}, \bibinfo {author} {\bibfnamefont {T.}~\bibnamefont {Wang}}, \bibinfo {author} {\bibfnamefont {S.~J.}\ \bibnamefont {Garratt}},\ and\ \bibinfo {author} {\bibfnamefont {E.}~\bibnamefont {Altman}},\ }\href {https://arxiv.org/abs/2510.08490} {\bibinfo {title} {Fractional quantum hall states under density decoherence}} (\bibinfo {year} {2025}{\natexlab{b}}),\ \Eprint {https://arxiv.org/abs/2510.08490} {arXiv:2510.08490 [cond-mat.str-el]} \BibitemShut {NoStop}%
\bibitem [{\citenamefont {Sommers}\ \emph {et~al.}(2026)\citenamefont {Sommers}, \citenamefont {Jacoby}, \citenamefont {Weinstein}, \citenamefont {Huse},\ and\ \citenamefont {Gopalakrishnan}}]{d4wh-hqcp}%
  \BibitemOpen
  \bibfield  {author} {\bibinfo {author} {\bibfnamefont {G.~M.}\ \bibnamefont {Sommers}}, \bibinfo {author} {\bibfnamefont {J.~A.}\ \bibnamefont {Jacoby}}, \bibinfo {author} {\bibfnamefont {Z.}~\bibnamefont {Weinstein}}, \bibinfo {author} {\bibfnamefont {D.~A.}\ \bibnamefont {Huse}},\ and\ \bibinfo {author} {\bibfnamefont {S.}~\bibnamefont {Gopalakrishnan}},\ }\bibfield  {title} {\bibinfo {title} {Spectral properties and coding transitions of haar-random quantum codes},\ }\href {https://doi.org/10.1103/d4wh-hqcp} {\bibfield  {journal} {\bibinfo  {journal} {PRX Quantum}\ }\textbf {\bibinfo {volume} {7}},\ \bibinfo {pages} {020328} (\bibinfo {year} {2026})}\BibitemShut {NoStop}%
\bibitem [{\citenamefont {Dennis}\ \emph {et~al.}(2002)\citenamefont {Dennis}, \citenamefont {Kitaev}, \citenamefont {Landahl},\ and\ \citenamefont {Preskill}}]{DennisKitaevLandahlPreskill2002}%
  \BibitemOpen
  \bibfield  {author} {\bibinfo {author} {\bibfnamefont {E.}~\bibnamefont {Dennis}}, \bibinfo {author} {\bibfnamefont {A.}~\bibnamefont {Kitaev}}, \bibinfo {author} {\bibfnamefont {A.}~\bibnamefont {Landahl}},\ and\ \bibinfo {author} {\bibfnamefont {J.}~\bibnamefont {Preskill}},\ }\bibfield  {title} {\bibinfo {title} {Topological quantum memory},\ }\href {https://doi.org/10.1063/1.1499754} {\bibfield  {journal} {\bibinfo  {journal} {Journal of Mathematical Physics}\ }\textbf {\bibinfo {volume} {43}},\ \bibinfo {pages} {4452} (\bibinfo {year} {2002})},\ \Eprint {https://arxiv.org/abs/quant-ph/0110143} {arXiv:quant-ph/0110143} \BibitemShut {NoStop}%
\bibitem [{\citenamefont {Wang}\ \emph {et~al.}(2003)\citenamefont {Wang}, \citenamefont {Harrington},\ and\ \citenamefont {Preskill}}]{Wang_2003}%
  \BibitemOpen
  \bibfield  {author} {\bibinfo {author} {\bibfnamefont {C.}~\bibnamefont {Wang}}, \bibinfo {author} {\bibfnamefont {J.}~\bibnamefont {Harrington}},\ and\ \bibinfo {author} {\bibfnamefont {J.}~\bibnamefont {Preskill}},\ }\bibfield  {title} {\bibinfo {title} {Confinement-higgs transition in a disordered gauge theory and the accuracy threshold for quantum memory},\ }\href {https://doi.org/10.1016/s0003-4916(02)00019-2} {\bibfield  {journal} {\bibinfo  {journal} {Annals of Physics}\ }\textbf {\bibinfo {volume} {303}},\ \bibinfo {pages} {31–58} (\bibinfo {year} {2003})}\BibitemShut {NoStop}%
\bibitem [{\citenamefont {Bombin}\ \emph {et~al.}(2012)\citenamefont {Bombin}, \citenamefont {Andrist}, \citenamefont {Ohzeki}, \citenamefont {Katzgraber},\ and\ \citenamefont {Martin-Delgado}}]{Bombin2012}%
  \BibitemOpen
  \bibfield  {author} {\bibinfo {author} {\bibfnamefont {H.}~\bibnamefont {Bombin}}, \bibinfo {author} {\bibfnamefont {R.~S.}\ \bibnamefont {Andrist}}, \bibinfo {author} {\bibfnamefont {M.}~\bibnamefont {Ohzeki}}, \bibinfo {author} {\bibfnamefont {H.~G.}\ \bibnamefont {Katzgraber}},\ and\ \bibinfo {author} {\bibfnamefont {M.~A.}\ \bibnamefont {Martin-Delgado}},\ }\bibfield  {title} {\bibinfo {title} {Strong resilience of topological codes to depolarization},\ }\href {https://doi.org/10.1103/PhysRevX.2.021004} {\bibfield  {journal} {\bibinfo  {journal} {Phys. Rev. X}\ }\textbf {\bibinfo {volume} {2}},\ \bibinfo {pages} {021004} (\bibinfo {year} {2012})}\BibitemShut {NoStop}%
\bibitem [{\citenamefont {Chubb}\ and\ \citenamefont {Flammia}(2021)}]{Chubb_2021}%
  \BibitemOpen
  \bibfield  {author} {\bibinfo {author} {\bibfnamefont {C.~T.}\ \bibnamefont {Chubb}}\ and\ \bibinfo {author} {\bibfnamefont {S.~T.}\ \bibnamefont {Flammia}},\ }\bibfield  {title} {\bibinfo {title} {Statistical mechanical models for quantum codes with correlated noise},\ }\href {https://doi.org/10.4171/aihpd/105} {\bibfield  {journal} {\bibinfo  {journal} {Annales de l’Institut Henri Poincaré D, Combinatorics, Physics and their Interactions}\ }\textbf {\bibinfo {volume} {8}},\ \bibinfo {pages} {269–321} (\bibinfo {year} {2021})}\BibitemShut {NoStop}%
\bibitem [{\citenamefont {Schumacher}(1996)}]{PhysRevA.54.2614}%
  \BibitemOpen
  \bibfield  {author} {\bibinfo {author} {\bibfnamefont {B.}~\bibnamefont {Schumacher}},\ }\bibfield  {title} {\bibinfo {title} {Sending entanglement through noisy quantum channels},\ }\href {https://doi.org/10.1103/PhysRevA.54.2614} {\bibfield  {journal} {\bibinfo  {journal} {Phys. Rev. A}\ }\textbf {\bibinfo {volume} {54}},\ \bibinfo {pages} {2614} (\bibinfo {year} {1996})}\BibitemShut {NoStop}%
\bibitem [{\citenamefont {Schumacher}\ and\ \citenamefont {Nielsen}(1996)}]{SchumacherNielsen1996}%
  \BibitemOpen
  \bibfield  {author} {\bibinfo {author} {\bibfnamefont {B.}~\bibnamefont {Schumacher}}\ and\ \bibinfo {author} {\bibfnamefont {M.~A.}\ \bibnamefont {Nielsen}},\ }\bibfield  {title} {\bibinfo {title} {Quantum data processing and error correction},\ }\href {https://doi.org/10.1103/PhysRevA.54.2629} {\bibfield  {journal} {\bibinfo  {journal} {Physical Review A}\ }\textbf {\bibinfo {volume} {54}},\ \bibinfo {pages} {2629} (\bibinfo {year} {1996})}\BibitemShut {NoStop}%
\bibitem [{\citenamefont {Lloyd}(1997)}]{PhysRevA.55.1613}%
  \BibitemOpen
  \bibfield  {author} {\bibinfo {author} {\bibfnamefont {S.}~\bibnamefont {Lloyd}},\ }\bibfield  {title} {\bibinfo {title} {Capacity of the noisy quantum channel},\ }\href {https://doi.org/10.1103/PhysRevA.55.1613} {\bibfield  {journal} {\bibinfo  {journal} {Phys. Rev. A}\ }\textbf {\bibinfo {volume} {55}},\ \bibinfo {pages} {1613} (\bibinfo {year} {1997})}\BibitemShut {NoStop}%
\bibitem [{\citenamefont {Barnum}\ and\ \citenamefont {Knill}(2002)}]{BarnumKnill2002}%
  \BibitemOpen
  \bibfield  {author} {\bibinfo {author} {\bibfnamefont {H.}~\bibnamefont {Barnum}}\ and\ \bibinfo {author} {\bibfnamefont {E.}~\bibnamefont {Knill}},\ }\bibfield  {title} {\bibinfo {title} {Reversing quantum dynamics with near-optimal quantum and classical fidelity},\ }\href {https://doi.org/10.1063/1.1459754} {\bibfield  {journal} {\bibinfo  {journal} {Journal of Mathematical Physics}\ }\textbf {\bibinfo {volume} {43}},\ \bibinfo {pages} {2097} (\bibinfo {year} {2002})},\ \Eprint {https://arxiv.org/abs/quant-ph/0004088} {arXiv:quant-ph/0004088} \BibitemShut {NoStop}%
\bibitem [{\citenamefont {B\'{e}ny}\ and\ \citenamefont {Oreshkov}(2010)}]{BenyOreshkov2010}%
  \BibitemOpen
  \bibfield  {author} {\bibinfo {author} {\bibfnamefont {C.}~\bibnamefont {B\'{e}ny}}\ and\ \bibinfo {author} {\bibfnamefont {O.}~\bibnamefont {Oreshkov}},\ }\bibfield  {title} {\bibinfo {title} {General conditions for approximate quantum error correction and near-optimal recovery channels},\ }\href {https://doi.org/10.1103/PhysRevLett.104.120501} {\bibfield  {journal} {\bibinfo  {journal} {Physical Review Letters}\ }\textbf {\bibinfo {volume} {104}},\ \bibinfo {pages} {120501} (\bibinfo {year} {2010})},\ \Eprint {https://arxiv.org/abs/0907.5391} {arXiv:0907.5391} \BibitemShut {NoStop}%
\bibitem [{\citenamefont {Colmenarez}\ \emph {et~al.}(2026)\citenamefont {Colmenarez}, \citenamefont {M\'{a}rton},\ and\ \citenamefont {M\"uller}}]{ColmenarezMartonMuller2026}%
  \BibitemOpen
  \bibfield  {author} {\bibinfo {author} {\bibfnamefont {L.}~\bibnamefont {Colmenarez}}, \bibinfo {author} {\bibfnamefont {A.}~\bibnamefont {M\'{a}rton}},\ and\ \bibinfo {author} {\bibfnamefont {M.}~\bibnamefont {M\"uller}},\ }\href@noop {} {\bibinfo {title} {Syndrome resampling enhances quantum error correction thresholds}} (\bibinfo {year} {2026}),\ \Eprint {https://arxiv.org/abs/2605.06101} {arXiv:2605.06101 [quant-ph]} \BibitemShut {NoStop}%
\bibitem [{\citenamefont {Smith}\ \emph {et~al.}(2024)\citenamefont {Smith}, \citenamefont {Brown},\ and\ \citenamefont {Bartlett}}]{Smith_2024}%
  \BibitemOpen
  \bibfield  {author} {\bibinfo {author} {\bibfnamefont {S.~C.}\ \bibnamefont {Smith}}, \bibinfo {author} {\bibfnamefont {B.~J.}\ \bibnamefont {Brown}},\ and\ \bibinfo {author} {\bibfnamefont {S.~D.}\ \bibnamefont {Bartlett}},\ }\bibfield  {title} {\bibinfo {title} {Mitigating errors in logical qubits},\ }\bibfield  {journal} {\bibinfo  {journal} {Communications Physics}\ }\textbf {\bibinfo {volume} {7}},\ \href {https://doi.org/10.1038/s42005-024-01883-4} {10.1038/s42005-024-01883-4} (\bibinfo {year} {2024})\BibitemShut {NoStop}%
\bibitem [{\citenamefont {English}\ \emph {et~al.}(2025)\citenamefont {English}, \citenamefont {Williamson},\ and\ \citenamefont {Bartlett}}]{English2025}%
  \BibitemOpen
  \bibfield  {author} {\bibinfo {author} {\bibfnamefont {L.~H.}\ \bibnamefont {English}}, \bibinfo {author} {\bibfnamefont {D.~J.}\ \bibnamefont {Williamson}},\ and\ \bibinfo {author} {\bibfnamefont {S.~D.}\ \bibnamefont {Bartlett}},\ }\bibfield  {title} {\bibinfo {title} {Thresholds for postselected quantum error correction from statistical mechanics},\ }\href {https://doi.org/10.1103/nh49-52y2} {\bibfield  {journal} {\bibinfo  {journal} {Phys. Rev. Lett.}\ }\textbf {\bibinfo {volume} {135}},\ \bibinfo {pages} {120603} (\bibinfo {year} {2025})}\BibitemShut {NoStop}%
\bibitem [{\citenamefont {Ahlswede}\ and\ \citenamefont {Gacs}(1976)}]{Ahlswede1976}%
  \BibitemOpen
  \bibfield  {author} {\bibinfo {author} {\bibfnamefont {R.}~\bibnamefont {Ahlswede}}\ and\ \bibinfo {author} {\bibfnamefont {P.}~\bibnamefont {Gacs}},\ }\bibfield  {title} {\bibinfo {title} {Spreading of sets in product spaces and hypercontraction of the markov operator},\ }\href {https://doi.org/10.1214/aop/1176995937} {\bibfield  {journal} {\bibinfo  {journal} {The Annals of Probability}\ }\textbf {\bibinfo {volume} {4}},\ \bibinfo {pages} {925} (\bibinfo {year} {1976})}\BibitemShut {NoStop}%
\bibitem [{\citenamefont {K{\"o}rner}\ and\ \citenamefont {Marton}(1977)}]{KornerMarton1977Comparison}%
  \BibitemOpen
  \bibfield  {author} {\bibinfo {author} {\bibfnamefont {J.}~\bibnamefont {K{\"o}rner}}\ and\ \bibinfo {author} {\bibfnamefont {K.}~\bibnamefont {Marton}},\ }\bibfield  {title} {\bibinfo {title} {Comparison of two noisy channels},\ }in\ \href@noop {} {\emph {\bibinfo {booktitle} {Topics in Information Theory}}},\ \bibinfo {series} {Colloquia Mathematica Societatis J{\'a}nos Bolyai}, Vol.~\bibinfo {volume} {16},\ \bibinfo {editor} {edited by\ \bibinfo {editor} {\bibfnamefont {I.}~\bibnamefont {Csisz{\'a}r}}\ and\ \bibinfo {editor} {\bibfnamefont {P.}~\bibnamefont {Elias}}}\ (\bibinfo  {publisher} {North-Holland},\ \bibinfo {address} {Amsterdam},\ \bibinfo {year} {1977})\ pp.\ \bibinfo {pages} {411--423},\ \bibinfo {note} {proceedings of the 2nd Colloquium on Information Theory, Keszthely, Hungary, 1975}\BibitemShut {NoStop}%
\bibitem [{\citenamefont {Raginsky}(2016)}]{raginsky2016strongdataprocessinginequalities}%
  \BibitemOpen
  \bibfield  {author} {\bibinfo {author} {\bibfnamefont {M.}~\bibnamefont {Raginsky}},\ }\href {https://arxiv.org/abs/1411.3575} {\bibinfo {title} {Strong data processing inequalities and sobolev inequalities for discrete channels}} (\bibinfo {year} {2016}),\ \Eprint {https://arxiv.org/abs/1411.3575} {arXiv:1411.3575 [cs.IT]} \BibitemShut {NoStop}%
\bibitem [{\citenamefont {Polyanskiy}\ and\ \citenamefont {Wu}(2017)}]{PolyanskiyWu2017}%
  \BibitemOpen
  \bibfield  {author} {\bibinfo {author} {\bibfnamefont {Y.}~\bibnamefont {Polyanskiy}}\ and\ \bibinfo {author} {\bibfnamefont {Y.}~\bibnamefont {Wu}},\ }\bibfield  {title} {\bibinfo {title} {Strong data-processing inequalities for channels and bayesian networks},\ }in\ \href {https://doi.org/10.1007/978-1-4939-7005-6_7} {\emph {\bibinfo {booktitle} {Convexity and Concentration}}},\ \bibinfo {series} {The IMA Volumes in Mathematics and its Applications}, Vol.\ \bibinfo {volume} {161},\ \bibinfo {editor} {edited by\ \bibinfo {editor} {\bibfnamefont {E.}~\bibnamefont {Carlen}}, \bibinfo {editor} {\bibfnamefont {M.}~\bibnamefont {Madiman}},\ and\ \bibinfo {editor} {\bibfnamefont {E.~M.}\ \bibnamefont {Werner}}}\ (\bibinfo  {publisher} {Springer},\ \bibinfo {year} {2017})\ pp.\ \bibinfo {pages} {211--249},\ \Eprint {https://arxiv.org/abs/1508.06025} {arXiv:1508.06025} \BibitemShut {NoStop}%
\bibitem [{\citenamefont {Makur}\ and\ \citenamefont {Polyanskiy}(2018)}]{Makur_2018}%
  \BibitemOpen
  \bibfield  {author} {\bibinfo {author} {\bibfnamefont {A.}~\bibnamefont {Makur}}\ and\ \bibinfo {author} {\bibfnamefont {Y.}~\bibnamefont {Polyanskiy}},\ }\bibfield  {title} {\bibinfo {title} {Comparison of channels: Criteria for domination by a symmetric channel},\ }\href {https://doi.org/10.1109/tit.2018.2839743} {\bibfield  {journal} {\bibinfo  {journal} {IEEE Transactions on Information Theory}\ }\textbf {\bibinfo {volume} {64}},\ \bibinfo {pages} {5704–5725} (\bibinfo {year} {2018})}\BibitemShut {NoStop}%
\bibitem [{\citenamefont {Harris}(1960)}]{Harris1960}%
  \BibitemOpen
  \bibfield  {author} {\bibinfo {author} {\bibfnamefont {T.~E.}\ \bibnamefont {Harris}},\ }\bibfield  {title} {\bibinfo {title} {A lower bound for the critical probability in a certain percolation process},\ }\href {https://doi.org/10.1017/S0305004100034241} {\bibfield  {journal} {\bibinfo  {journal} {Mathematical Proceedings of the Cambridge Philosophical Society}\ }\textbf {\bibinfo {volume} {56}},\ \bibinfo {pages} {13} (\bibinfo {year} {1960})}\BibitemShut {NoStop}%
\bibitem [{\citenamefont {Griffiths}(1967)}]{Griffiths}%
  \BibitemOpen
  \bibfield  {author} {\bibinfo {author} {\bibfnamefont {R.~B.}\ \bibnamefont {Griffiths}},\ }\bibfield  {title} {\bibinfo {title} {Correlations in ising ferromagnets. i},\ }\href {https://doi.org/10.1063/1.1705219} {\bibfield  {journal} {\bibinfo  {journal} {Journal of Mathematical Physics}\ }\textbf {\bibinfo {volume} {8}},\ \bibinfo {pages} {478} (\bibinfo {year} {1967})}\BibitemShut {NoStop}%
\bibitem [{\citenamefont {Kelly}\ and\ \citenamefont {Sherman}(1968)}]{KellySherman1968}%
  \BibitemOpen
  \bibfield  {author} {\bibinfo {author} {\bibfnamefont {D.~G.}\ \bibnamefont {Kelly}}\ and\ \bibinfo {author} {\bibfnamefont {S.}~\bibnamefont {Sherman}},\ }\bibfield  {title} {\bibinfo {title} {General griffiths' inequalities on correlations in ising ferromagnets},\ }\href {https://doi.org/10.1063/1.1664600} {\bibfield  {journal} {\bibinfo  {journal} {Journal of Mathematical Physics}\ }\textbf {\bibinfo {volume} {9}},\ \bibinfo {pages} {466} (\bibinfo {year} {1968})}\BibitemShut {NoStop}%
\bibitem [{\citenamefont {Ginibre}(1970)}]{Ginibre1970}%
  \BibitemOpen
  \bibfield  {author} {\bibinfo {author} {\bibfnamefont {J.}~\bibnamefont {Ginibre}},\ }\bibfield  {title} {\bibinfo {title} {General formulation of griffiths' inequalities},\ }\href {https://doi.org/10.1007/BF01646537} {\bibfield  {journal} {\bibinfo  {journal} {Communications in Mathematical Physics}\ }\textbf {\bibinfo {volume} {16}},\ \bibinfo {pages} {310} (\bibinfo {year} {1970})}\BibitemShut {NoStop}%
\bibitem [{\citenamefont {Calderbank}\ \emph {et~al.}(1998)\citenamefont {Calderbank}, \citenamefont {Rains}, \citenamefont {Shor},\ and\ \citenamefont {Sloane}}]{681315}%
  \BibitemOpen
  \bibfield  {author} {\bibinfo {author} {\bibfnamefont {A.}~\bibnamefont {Calderbank}}, \bibinfo {author} {\bibfnamefont {E.}~\bibnamefont {Rains}}, \bibinfo {author} {\bibfnamefont {P.}~\bibnamefont {Shor}},\ and\ \bibinfo {author} {\bibfnamefont {N.}~\bibnamefont {Sloane}},\ }\bibfield  {title} {\bibinfo {title} {Quantum error correction via codes over gf(4)},\ }\href {https://doi.org/10.1109/18.681315} {\bibfield  {journal} {\bibinfo  {journal} {IEEE Transactions on Information Theory}\ }\textbf {\bibinfo {volume} {44}},\ \bibinfo {pages} {1369} (\bibinfo {year} {1998})}\BibitemShut {NoStop}%
\bibitem [{\citenamefont {Gidney}(2021)}]{Gidney_2021}%
  \BibitemOpen
  \bibfield  {author} {\bibinfo {author} {\bibfnamefont {C.}~\bibnamefont {Gidney}},\ }\bibfield  {title} {\bibinfo {title} {Stim: a fast stabilizer circuit simulator},\ }\href {https://doi.org/10.22331/q-2021-07-06-497} {\bibfield  {journal} {\bibinfo  {journal} {Quantum}\ }\textbf {\bibinfo {volume} {5}},\ \bibinfo {pages} {497} (\bibinfo {year} {2021})}\BibitemShut {NoStop}%
\bibitem [{nba()}]{nbalance}%
  \BibitemOpen
  \href@noop {} {}\bibinfo {note} {For general quantum noise, $n$-balance is equivalent to $\cP_n^\dagger(I_{SQ})=Z_nI_{\mathrm L}$ and need not hold; see Appendix.~B.}\BibitemShut {Stop}%
\bibitem [{tri()}]{trivial-qubit-check}%
  \BibitemOpen
  \href@noop {} {}\bibinfo {note} {For the trivial one-qubit code under depolarizing noise with $p=0.3$, $I_c^{(2)}\approx0.039$, while $I_{c,\mathrm{match}}=I_c^{(1)}\approx-0.247$.}\BibitemShut {Stop}%
\bibitem [{\citenamefont {Vijay}\ \emph {et~al.}(2026)\citenamefont {Vijay}, \citenamefont {Raj}, \citenamefont {Kudler-Flam}, \citenamefont {Vermersch}, \citenamefont {Elben},\ and\ \citenamefont {Nie}}]{Vijay_SAC}%
  \BibitemOpen
  \bibfield  {author} {\bibinfo {author} {\bibfnamefont {A.}~\bibnamefont {Vijay}}, \bibinfo {author} {\bibfnamefont {A.}~\bibnamefont {Raj}}, \bibinfo {author} {\bibfnamefont {J.}~\bibnamefont {Kudler-Flam}}, \bibinfo {author} {\bibfnamefont {B.}~\bibnamefont {Vermersch}}, \bibinfo {author} {\bibfnamefont {A.}~\bibnamefont {Elben}},\ and\ \bibinfo {author} {\bibfnamefont {L.}~\bibnamefont {Nie}},\ }\bibfield  {title} {\bibinfo {title} {Analytically continuing the randomized measurement toolbox},\ }\href {https://doi.org/10.1103/837t-68bf} {\bibfield  {journal} {\bibinfo  {journal} {Phys. Rev. Lett.}\ }\textbf {\bibinfo {volume} {137}},\ \bibinfo {pages} {100202} (\bibinfo {year} {2026})}\BibitemShut {NoStop}%
\bibitem [{\citenamefont {Fletcher}\ \emph {et~al.}(2008)\citenamefont {Fletcher}, \citenamefont {Shor},\ and\ \citenamefont {Win}}]{FletcherShorWin2008}%
  \BibitemOpen
  \bibfield  {author} {\bibinfo {author} {\bibfnamefont {A.~S.}\ \bibnamefont {Fletcher}}, \bibinfo {author} {\bibfnamefont {P.~W.}\ \bibnamefont {Shor}},\ and\ \bibinfo {author} {\bibfnamefont {M.~Z.}\ \bibnamefont {Win}},\ }\bibfield  {title} {\bibinfo {title} {Channel-adapted quantum error correction for the amplitude damping channel},\ }\href {https://doi.org/10.1109/TIT.2008.2006453} {\bibfield  {journal} {\bibinfo  {journal} {IEEE Transactions on Information Theory}\ }\textbf {\bibinfo {volume} {54}},\ \bibinfo {pages} {5705} (\bibinfo {year} {2008})},\ \Eprint {https://arxiv.org/abs/0710.1052} {arXiv:0710.1052} \BibitemShut {NoStop}%
\bibitem [{\citenamefont {Fuchs}\ and\ \citenamefont {van~de Graaf}(1999)}]{FuchsVanDeGraaf1999}%
  \BibitemOpen
  \bibfield  {author} {\bibinfo {author} {\bibfnamefont {C.~A.}\ \bibnamefont {Fuchs}}\ and\ \bibinfo {author} {\bibfnamefont {J.}~\bibnamefont {van~de Graaf}},\ }\bibfield  {title} {\bibinfo {title} {Cryptographic distinguishability measures for quantum-mechanical states},\ }\href@noop {} {\bibfield  {journal} {\bibinfo  {journal} {IEEE Trans. Inf. Theory}\ }\textbf {\bibinfo {volume} {45}},\ \bibinfo {pages} {1216} (\bibinfo {year} {1999})}\BibitemShut {NoStop}%
\bibitem [{\citenamefont {Choi}(1975)}]{Choi1975}%
  \BibitemOpen
  \bibfield  {author} {\bibinfo {author} {\bibfnamefont {M.-D.}\ \bibnamefont {Choi}},\ }\bibfield  {title} {\bibinfo {title} {Completely positive linear maps on complex matrices},\ }\href@noop {} {\bibfield  {journal} {\bibinfo  {journal} {Linear Algebra Appl.}\ }\textbf {\bibinfo {volume} {10}},\ \bibinfo {pages} {285} (\bibinfo {year} {1975})}\BibitemShut {NoStop}%
\bibitem [{\citenamefont {Jamio{\l}kowski}(1972)}]{Jamiolkowski1972}%
  \BibitemOpen
  \bibfield  {author} {\bibinfo {author} {\bibfnamefont {A.}~\bibnamefont {Jamio{\l}kowski}},\ }\bibfield  {title} {\bibinfo {title} {Linear transformations which preserve trace and positive semidefiniteness of operators},\ }\href@noop {} {\bibfield  {journal} {\bibinfo  {journal} {Rep. Math. Phys.}\ }\textbf {\bibinfo {volume} {3}},\ \bibinfo {pages} {275} (\bibinfo {year} {1972})}\BibitemShut {NoStop}%
\bibitem [{\citenamefont {Uhlmann}(1976)}]{Uhlmann1976}%
  \BibitemOpen
  \bibfield  {author} {\bibinfo {author} {\bibfnamefont {A.}~\bibnamefont {Uhlmann}},\ }\bibfield  {title} {\bibinfo {title} {The ``transition probability'' in the state space of a $*$-algebra},\ }\href@noop {} {\bibfield  {journal} {\bibinfo  {journal} {Rep. Math. Phys.}\ }\textbf {\bibinfo {volume} {9}},\ \bibinfo {pages} {273} (\bibinfo {year} {1976})}\BibitemShut {NoStop}%
\end{thebibliography}
\end{document}